\documentclass[submission,copyright,creativecommons]{eptcs}
\providecommand{\event}{} 

\usepackage{amsmath}
\usepackage{amsthm}
\usepackage{stmaryrd}
\usepackage{amssymb}
\usepackage{mathpartir}
\usepackage{framed}

\newtheorem{theorem}{Theorem}[section]
\newtheorem{proposition}[theorem]{\textbf{Proposition}}
\newtheorem{lemma}[theorem]{\textbf{Lemma}}
\newtheorem{corollary}[]{\textbf{Corollary}}
\newtheorem{example}{Example}[section]
\usepackage{thmtools, thm-restate}

\newcommand{\ie}[0]{\textit{i.e.}}

\newcommand{\wrt}[0]{\textit{w.r.t.}}
\newcommand{\st}[0]{\textit{s.t.}}

\newcommand{\tom}{\textsf{Tom}}
\newcommand{\maude}{\textsf{Maude}}
\newcommand{\aprove}{\textsf{AProVE}}
\newcommand{\TTT}{\textsf{TTT2}}
\newcommand{\java}{\textsf{Java}}

\newcommand{\otrsep}{,}
\DeclareMathAlphabet{\mathcal}{OMS}{cmsy}{m}{n}
\newcommand{\FF}[0]{\mathcal{F}}
\newcommand{\XX}[0]{\mathcal{X}}
\newcommand{\RR}[0]{\mathcal{R}}
\newcommand{\CC}[0]{\mathcal{C}}
\newcommand{\CCt}[0]{\mathcal{L}}
\newcommand{\DD}[0]{\mathcal{D}}
\newcommand{\TT}[2]{\mathcal{T}(#1,#2)}
\newcommand{\EE}[0]{\mathcal{E}}

\newcommand{\LL}[0]{\mathcal{L}}
\newcommand{\PP}[0]{P}

\newcommand{\TF}[0]{\mathcal{T}({\FF})}
\newcommand{\TFX}{\TT{\FF}{\XX}}
\newcommand{\TC}[0]{\mathcal{T}({\CC})}
\newcommand{\TCX}{\TT{\CC}{\XX}}
\newcommand{\ETCX}{{\mathcal{T}_\mathcal{E}(\CC,\XX)}}
\newcommand{\ETCXap}{\mathcal{T}_\mathcal{E}(\CC\cup\{\ap\},\XX)}

\newcommand{\PPos}[0]{\mathcal{P}os}

\newcommand{\stt}[2]{\ensuremath{#1_{|#2}}}
\newcommand{\rmp}[3]{\ensuremath{{#1\left[#3\right]_{#2}}}}
\newcommand{\var}[1]{\mathcal{V}ar\left({#1}\right)}
\newcommand{\fvar}[1]{\mathcal{F}\mathcal{V}ar\left({#1}\right)}
\newcommand{\mvar}[1]{\mathcal{M}\mathcal{V}ar\left({#1}\right)}
\newcommand{\dom}[1]{\mathcal{D}om\left({#1}\right)}
\newcommand{\ar}[0]{ar}
\newcommand{\ra}[0]{\rightarrowtriangle}
\newcommand{\raM}[0]{\Rightarrow}
\newcommand{\ovr}[1]{\overrightarrow{#1}}

\DeclareRobustCommand\longtwoheadleftarrow
     {\mathrel{\longleftarrow\mkern-26mu\leftarrow~} }
\DeclareRobustCommand\longtwoheadrightarrow
     {\mathrel{~\rightarrow\mkern-26mu\longrightarrow} }

\newcommand{\multievalM}[1]{\evalM_{#1}^{*}}
\newcommand{\eval}[0]  {\longrightarrow}
\newcommand{\evalM}[0]  {\Longrightarrow}
\newcommand{\evaleps}[1]  {\stackrel{#1}{\eval}}

\newcommand{\plaint}{{term}}
\newcommand{\prefix}[2]{{#1}<{#2}}

\newcommand{\tupsymbol}{\cdot}
\newcommand{\tup}[1]{\llparenthesis #1 \rrparenthesis}
\newcommand{\foo}{{\varphi}}

\newcommand{\valuet}{{value}}
\newcommand{\addt}{{additive}}

\newcommand{\minus}{\setminus}%
\newcommand{\plus}{+}
\newcommand{\semg}[1]{\llbracket #1 \rrbracket}

\newcommand{\matcha}     {\mathrel{\mbox{$\prec\hspace{-0.4em}\prec$}}}
\newcommand{\nmatcha}     {\mathrel{\mbox{$\prec\hspace{-0.5em}\not\prec$}}}

\newcommand{\Rr}[0]{\mathfrak{R}}
\newcommand{\RrAP}[0]{\Rr_{\ap}}
\newcommand{\RrC}[0]{\Rr_{\minus}}
\newcommand{\RrCat}[0]{\Rr_{\minus}^{\at}}

\newcommand{\rmAddL}[0]{A1}
\newcommand{\rmAddR}[0]{A2}
\newcommand{\rmEmptyF}[0]{E1}
\newcommand{\rmEmptyA}[0]{E2}
\newcommand{\rmMinusMV}[0]{M1}
\newcommand{\rmMinusMB}[0]{M2}
\newcommand{\rmMinusMP}[0]{M3}
\newcommand{\rmMinusVM}[0]{M4}
\newcommand{\rmMinusVMp}[0]{\rmMinusVM'}
\newcommand{\rmMinusBM}[0]{M5}
\newcommand{\rmMinusPM}[0]{M6}
\newcommand{\rmMinusFG}[0]{M8}
\newcommand{\rmMinusFF}[0]{M7}
\newcommand{\rmMinusAM}[0]{M9}
\newcommand{\rmMinusMA}[0]{M10}
\newcommand{\rmDistAdd}[0]{S1}
\newcommand{\rmDistAt}[0]{S2}

\newcommand{\metaLevel}[1]{\overline{#1}}
\newcommand{\meta}[1]{\metaLevel{#1}}

\newcommand{\varCt}[0]{z}
\newcommand{\varVar}[0]{\metaLevel{V}}

\newcommand{\tm}[0]{{t}}

\newcommand{\vm}[0]{{v}}
\newcommand{\wm}[0]{{w}}
\newcommand{\ap}[0]{\mathbf{!}}
\newcommand{\at}[0]{\mathbin{@}}

\newcommand{\antit}{{anti-term}}

\newcommand{\TRord}[0]{\mathfrak{T}^{<}}

\newcommand{\TRcomp}[0]{\mathfrak{T}^{\minus}}
\newcommand{\TRat}[0]{\mathfrak{T}^{\at}}
\newcommand{\TRall}[0]{\mathfrak{T^{ap}}}
\newcommand{\TRorder}[0]{\mathfrak{T}}

\newcommand{\semgset}[1]{\semg{#1}}

\newcommand{\validsubset}[2]{{#1}~\text{is a valid subset of}~{#2}}

\newcommand{\minimum}[0]{minimum}
\newcommand{\Smin}[0]{S^{m}}
\newcommand{\card}[1]{\vert #1 \vert}

\newcommand{\erw}{extended rewrite rule}
\newcommand{\er}{extended rule}

\newcommand{\rpo}{\prec}

\newcommand{\sizeorder}{\supset}
\newcommand{\varp}{\mathfrak{pos}}
\newcommand{\symp}{\mathfrak{sym}}
\newcommand{\sizep}{\mathfrak{size}}
\newcommand{\patvar}{\mathfrak{var}}

\newcommand{\allPos}{\mathfrak{Pos}}

\title{Generic Encodings of Constructor Rewriting Systems}
\author{Horatiu Cirstea
\institute{Universit{\'e} de Lorraine -- LORIA\\
}
\email{Horatiu.Cirstea@loria.fr}
\and
Pierre-Etienne Moreau
\institute{Universit{\'e} de Lorraine -- LORIA\\
}
\email{Pierre-Etienne.Moreau@loria.fr}
}

\def\titlerunning{Generic Encodings of Constructor Rewriting Systems}
\def\authorrunning{H. Cirstea and P.-E. Moreau}

\begin{document}
\maketitle 
\begin{abstract}
Rewriting is a formalism widely used in computer science and
mathematical logic. The classical formalism has been extended, in the
context of functional languages, with an order over the rules and, in
the context of rewrite based languages, with the negation over
patterns.
We propose in this paper a concise and clear algorithm computing the
difference over patterns which can be used to define generic encodings
of constructor term rewriting systems with negation and order into
classical term rewriting systems.
As a direct consequence, established methods used for term rewriting
systems can be applied to analyze properties of the extended systems.
The approach can also be seen as a generic compiler which targets any
language providing basic pattern matching primitives.
The formalism provides also a new method for deciding if a set of
patterns subsumes a given pattern and thus, for checking the presence
of useless patterns or the completeness of a set of patterns.

\end{abstract}

\section{Introduction}
\label{se:intro}
%
Rewriting is a very powerful tool used in theoretical studies as well
as for practical implementations. It is used, for example, in
semantics in order to describe the meaning of programming
languages, but also in automated reasoning when describing by
inference rules a logic, a theorem prover or a constraint
solver. It is also used to compute in systems making the notion of
rule an explicit and first class object, like Mathematica~\cite{IMS-2015-Marin},
{\maude}~\cite{Maude2:03}, or
{\tom}~\cite{BallandBKMR-RTA2007}.  Rewrite rules, the core concept in
rewriting, consist of a pattern that describes a schematic situation
and the transformation that should be applied in that particular
case. The pattern expresses a potentially infinite number of instances
and the application of the rewrite rule is decided locally using a
(matching) algorithm which only depends on the pattern and its
subject.

Comparing to the general rewriting formalism where rule application is
decided locally and independently of the other rules, rule-based and
functional programming languages generally use an order over the
rules. This is not only convenient from the implementation point of
view but it also allows more concise and clear specifications in some
specific cases. In particular, this order might avoid an exhaustive
specification of alternative and default cases.
For instance, if we consider a term representation of motor vehicles
we can use the following list of rules
\[\arraycolsep=1.4pt
\begin{array}{ll}
[
  &
  paint(car(x,suv)) \ra red 
  \otrsep
  \\
  &  
  paint(car(electric,x)) \ra blue
  \otrsep
  \\
  &
  paint(car(diesel,y)) \ra red
  \otrsep
  \\
  &
  paint(car(x,y)) \ra white 
  \otrsep
  \\
  &
  paint(x)  \ra  red
~~]
\end{array}
\]
for the assignment of an imaginary eco-label: all electric cars but
the SUVs (which are red) are blue, diesel cars are red and the
remaining cars are white; all the other vehicles are red.

Patterns express positive conditions and we have used the term
$car(electric,x)$ to specify electric cars of any style. Negation is
nevertheless intrinsic to human thinking and most of the time when
searching for something, we base our patterns on both positive and
negative conditions. We would like for example to specify that we
search for all cars that are not SUVs, or for all cars which are
neither SUV nor diesel.
The notion of pattern has been extended to the one of
anti-pattern~\cite{KirchnerKM-2007}, i.e. patterns that may contain
complement symbols, and implemented in tools featuring pattern
matching like~{\tom}~\cite{jsc2010} and
Mathematica~\cite{IMS-2015-Marin}. With such an approach the above
statements can be easily expressed as $car(x, \ap suv)$ and
respectively $car(\ap diesel,\ap suv)$, and the eco-labeling can be
expressed by the following list of rules with anti-patterns
%
\[\arraycolsep=1.4pt
\begin{array}{ll}
[
&
  paint(car(electric,!suv)) \ra blue 
  \otrsep
  \\
  &
  paint(car(!diesel,!suv))) \ra white
  \otrsep
  \\
  &
  paint(x)  \ra  red 
~~]
\end{array}
\]

Similarly to plain term rewriting systems (TRS), {\ie} TRS without
anti-patterns and ordered rules, it is interesting to analyze the
extended systems {\wrt} to their confluence, termination and
reachability properties, for example.  Generally, well-established
techniques and (automatic) tools used in the plain case cannot be
applied directly in the general case.
There have been several works in the context of functional programming
like, for
example~\cite{KraussSTFG11,toplas-GieslRSST11,GieslBEFFOPSSST14,AvanziniLM15}
to cite only a few, but they are essentially focused on powerful
techniques for analyzing the termination and complexity of functional
programs with ordered matching statements.
We are interested here in a transformation approach which can be used
as an add-on for well-established analyzing techniques and tools but
also as a generic compiler for ordered TRS involving anti-patterns
which could be easily integrated in any language providing rewrite
rules, or at least pattern matching primitives.
For example, if we consider trucks and cars with 
$4$ fuel types and $3$ styles the transformation  we propose will
provide the following 
order independent set of rules:
\[\arraycolsep=1.4pt
\begin{array}{rl}
\{
  &
  paint(car(electric,sedan)) \ra blue
  \otrsep
  \\
  &
  paint(car(electric,minivan)) \ra blue
  \otrsep
  \\
  &
  paint(car(hybrid,sedan)) \ra white
  \otrsep
  \\
  &
  paint(car(hybrid,minivan)) \ra white
  \otrsep
  \\
  &
  paint(car(gas,sedan)) \ra white
  \otrsep
  \\
  &
  paint(car(gas,minivan)) \ra white
  \otrsep
  \\
  &
  paint(truck(x,y)) \ra red
  \otrsep
  \\
  &
  paint(car(x,suv)) \ra red
  \otrsep
  \\
  &
  paint(car(diesel,x)) \ra red
~~~~\}
\end{array}
\]
%
for the previous list of rules.

In this paper we propose an extended matching and rewriting formalism
which strongly relies on the newly introduced operation of relative
complement, and we provide an algorithm which computes for a given
difference of patterns $p_1\minus p_2$ the set of patterns which match
all terms matched by $p_1$ but those matched by $p_2$.  The algorithm
defined itself by rewriting in a concise and clear way turns out to be
not only easy to implement but also very powerful since it has several
direct applications:

\begin{itemize}
\item it can be used to transform an ordered constructor TRS into a
  plain constructor TRS
  defining exactly the same relation over terms;
\item it can be used to transform an anti-pattern into a set of
  equivalent patterns and provides thus a way to compile such
  patterns and to prove, using existing techniques, properties of
  anti-patterns and of the corresponding rewriting systems;
\item it can be used to decide whether a pattern is subsumed by a
  given set of patterns and thus, to check the presence of useless
  patterns or the completeness of a set of patterns.
\end{itemize}

The paper is organized as follows. The next section introduces the
notions of pattern, pattern semantics and rewriting
system. Section~\ref{se:encoding} presents the translation of extended
patterns into plain patterns and explains how this can be used to
detect useless patterns.  In Section~\ref{se:elimination} we present
a new technique for eliminating redundant patterns and
Section~\ref{se:caseEncode} describes the transformation of ordered
CTRS involving anti-patterns into plain CTRS.
Section~\ref{se:experiments} presents some optimizations and
implementation details.  In Section~\ref{se:related} we discuss some
related works.  We end with conclusions and further work.
%

\newcommand{\Gd}[1]{\mathcal{G}S(#1)}
\section{Pattern semantics and term rewriting systems}
\label{se:otrsdef}
%
We define in this section most of the notions and notations necessary
in the rest of the paper.

\subsection{Term rewriting systems}
We first briefly recall basic notions concerning first order terms and
term rewriting systems; more details can be found in
\cite{BaaderN98,Terese2002}.

A \emph{signature} $\Sigma$ consists in an alphabet $\FF$ of symbols
together with an application $\ar$ which associates to any symbol $f$
its \emph{arity} (we write $\FF^n$ for the subset of symbols of arity
$n$).  Symbols in $\FF^0$ are called \emph{constants}.
Given a countable set $\XX$ of \emph{variable} symbols, the set of
\emph{{\plaint}s} \emph{$\TFX$} is the smallest set containing $\XX$
and such that $f(t_1,\ldots,t_n)$ is in $\TFX$ whenever $f\in\FF^n$
and $t_i\in\TFX$ for $i\in [1,n]$.

A {\em position} of a term $t$ is a finite sequence of positive
integers describing the path from the root of $t$ to the root of the
sub-term at that position.  The empty sequence representing the root
position is denoted by $\varepsilon$.  $\stt t \omega$,
resp. $t(\omega)$, denotes the sub-term of $t$, resp. the symbol of
$t$, at position $\omega$. We denote by $\rmp t \omega s$ the term $t$
with the sub-term at position $\omega$ replaced by $s$. $\PPos(t)$ is
called the set of positions of $t$.
We write $\prefix{\omega_1}{\omega_2}$ if $\omega_2$ extends
$\omega_1$, that is, if $\omega_2=\omega_1.\omega_1'$ for some non
empty sequence $\omega_1'$.  We have thus,
$\prefix{\varepsilon}{\varepsilon.1}$ and
$\prefix{\varepsilon.1}{\varepsilon.1.2}$.  Notice that $\forall
\omega_1,\omega_2\in\PPos(t)$, $\prefix{\omega_1}{\omega_2}$ iff
$\stt{t}{\omega_2}$ is a sub-term of $\stt{t}{\omega_1}$.

The set of variables occuring in $t\in\TFX$ is denoted by $\var t$.
If $\var t$ is empty, $t$ is called a \emph{ground} term. $\TF$
denotes the set of all ground {\plaint}s.  A \emph{linear} term is a
term where every variable occurs at most once.

We call \emph{substitution} any mapping from $\XX$ to $\TFX$ which is
the identity except over a finite set of variables $\dom \sigma$
called \emph{domain} of $\sigma$.  A substitution $\sigma$ extends as
expected to an endomorphism $\sigma'$ of $\TFX$.  To simplify the
notations, we do not make the distinction between $\sigma$ and
$\sigma'$.  $\sigma$ is often denoted by $\{x \mapsto
\sigma(x)~|~x\in\dom \sigma\}$.

A \emph{rewrite rule} (over $\Sigma$) is a pair
$(l,r)\in\TFX\times\TFX$ (also denoted $l \raM r$) such that $\var
r\subseteq \var l$ and a \emph{term rewriting system} (TRS) is a set
of rewrite rules $\RR$ inducing a \emph{rewriting relation} over
$\TF$, denoted by $\evalM_{\RR}$ and such that $t
\evalM_{\RR} t'$ iff there exist $l\raM r\in\RR$,
$\omega\in\PPos(t)$, and a substitution $\sigma$ such that $\stt t
\omega = \sigma(l)$ and $t' = \rmp t \omega {\sigma(r)}$.  The
reflexive and transitive closure of $\evalM_{\RR}$ is denoted
by $\multievalM{\RR}$.

A rewriting system $\RR$ is \emph{left-linear} if the left-hand sides
of all its rewrite rules are linear.  
$\RR$ is \emph{confluent} 
when for any terms $t,t_1,t_2$ {\st} $t\multievalM{R}{t_1}$ and
$t\multievalM{R}{t_2}$ there exists a term $u$ {\st} $t_1\multievalM{R}{u}$ and
$t_2\multievalM{R}{u}$.
$\RR$ is \emph{terminating} if there exists no infinite rewrite
sequence $t_1\evalM_{\RR}t_2\evalM_{\RR}\cdots$.  A
terminating and confluent rewriting system $\RR$ is called convergent;
for such systems the normal form of $t$ is denoted
$t\downarrow_{\RR}$.

For the purpose of presenting function definitions with an ML-style
pattern matching we consider that the set of symbols~$\FF$ of a
signature is partitioned into a set~$\DD$ of \emph{defined symbols}
and a set~$\CC$ of \emph{constructors}.  The linear terms over the
constructor signature \emph{$\TCX$} are called \emph{constructor
  patterns} and the ground constructor patterns in $\TC$ are called
\emph{\valuet}s.
A constructor TRS (CTRS) is a TRS whose rules have a left-hand side of
the form $\foo(l_1,\ldots,l_n) \ra r$ with $\foo\in\DD^n$ and
$l_i\in\TCX$.

\subsection{Patterns and their ground semantics}
\label{se:patternsANDmatching}
%
The definition of a function $\foo$ 
by a list of oriented equations of the form:
\[
\begin{array}{l@{\hspace{5pt}}lcl}
[ 
&
  \foo(p_1^1,\ldots,p_n^1)& \ra& t^1 \\
&
  &\vdots& \\
&
  \foo(p_1^m,\ldots,p_n^m)& \ra& t^m 
~~]
\end{array}
\]
corresponds thus to an ordered CTRS 
with $\foo\in\DD^n$, $p_i^j\in\TCX$, $t^j\in\TFX$.

When focusing on the underlying pattern matching for such functional
specifications the defined symbol in the left-hand side of the
equations 
only indicates the name of the defined function and only the
constructor terms are relevant for its definition.
We assume thus a set $\CCt=\{\tupsymbol_1,\ldots,\tupsymbol_n\}$ of
suitable symbols for n-tuples (the cardinality of $\CCt$ is the
maximum arity of the symbols in $\DD$), and for simplicity an n-tuple
$\tupsymbol_n(p_1,\ldots,p_n)$ is denoted $\tup{p_1,\ldots,p_n}$.  In
order to address the underlying pattern matching of a function
definition of the above form we consider the list of tuples of
patterns:
\[
\begin{array}{c}
\tup{p_1^1,\ldots,p_n^1}\\
\vdots \\
\tup{p_1^m,\ldots,p_n^m}
\end{array}
\]
All the tuples of patterns $\tup{p_1,\ldots,p_n}$ considered in this
paper are linear, {\ie} each $p_i$ is linear, and a variable can
appear in only one pattern $p_i$.
In what follows, we call constructor pattern a constructor pattern or
a tuple of constructor patterns.  We may use the notation $\ovr{p}$ to
denote explicitly a tuple of constructor patterns.
Similarly, we call value a term in $\TC$ or a tuple of such values and
we use the notation $\ovr{v}$ to denote explicitly tuples of values.
We also write $\foo(\ovr{p})$ to denote a term $\foo(p_1,\ldots,p_n)$,
$\foo\in\DD^n$, when there is no need to make explicit the terms
$p_1,\ldots,p_n$ in a given context.

Let $v$ be a value and $p$ be a constructor pattern ({\ie} a
constructor pattern or a tuple of constructor patterns), we say that
$v$ is an instance of $p$ when there exists a substitution $\sigma$
(extended to the notion of tuples) such that $v=\sigma(p)$ and in this
case we say that $p$ \emph{matches} $v$. Since $p$ is linear the
instance relation can be defined inductively:
\[
\begin{array}{rcll}
  x & \matcha  &  v  & x\in\XX \\
  c(p_1,\ldots,p_n) & \matcha  &  c(v_1,\ldots,v_n)  & \text{iff }\land_{i=1}^n p_i \matcha v_i, 
  c\in\CC\cup\CCt \\
\end{array}
\]

Given a list of patterns $\PP=[{p_1},\ldots,{p_n}]$ we say that $\PP$
matches a value $v$ with pattern $p_i$, denoted $\PP\matcha_{i}{v}$,
iff the following conditions hold:
\[
\begin{array}{ll}
  {p_i} \matcha {v} \\
  {p_j} \nmatcha {v}, & \forall j<i\\
\end{array}
\]
Note that if $\PP\matcha_{i}{v}$ then for all $j\not=i$, $\PP\nmatcha_{j}{v}$.

Several pattern matching properties can be expressed in this
context~\cite{JFP:977988}:
\begin{itemize}
\item a list of patterns $\PP$ is \emph{exhaustive} iff for all values
  $v$ there exists an $i$ such that $\PP\matcha_{i}{v}$,
\item a pattern $p_i\in\PP$ is \emph{useless} iff there does not exist
  a value $v$ such that $\PP\matcha_{i}{v}$.
\end{itemize}

Starting from the observation that a pattern can be interpreted as the
set of its instances we define the semantics of (lists of) patterns
and state the relationship to pattern matching.

The \emph{ground semantics} of a constructor pattern $p\in\TCX$ is the
set of all its ground constructor instances: $\semg{p} = \{ \sigma(p)
\mid
\sigma(p)\in\TC\}$. 
This extends as expected to tuples of constructor
patterns: $\semg{\tup{p_1,\ldots,p_n}} = \{
\tup{\sigma(p_1),\ldots,\sigma(p_n)} \mid
\sigma(p_1),\ldots,\sigma(p_n)\in\TC\}$.
Note that the ground semantics of a variable~$x$ is the set of all
possible ground patterns: $\semg{x}=\TC$, and since patterns are
linear we can use a recursive definition for the non variable
patterns:
\[\semg{c(p_1,\ldots,p_n)}=\{c(t_1,\ldots,t_n)\mid (t_1,\ldots,t_n)\in\semg{p_1}\times\ldots\times\semg{p_n}\},
\]
for all $c\in\CC\cup\CCt$.

\begin{restatable}[Instance relation vs. ground semantics]{prop}{subsumptionVSsemantics}
  \label{th:subsumptionVSsemantics}
  Given a pattern ${p}$ and a value ${v}$, ${v}\in\semg{{p}}$ iff ${p}
  \matcha {v}$.
\end{restatable}

The semantics of a set of patterns $P=\{p_1,\ldots,p_n\}$ or of a list
of patterns $P=[p_1,\ldots,p_n]$ is the union of the semantics of each
of the patterns: $\semg{P}=\bigcup_{i=1}^{n}\semg{p_i}$.
Note that given a value ${v}$, ${v}\in\semg{P}$ iff there exists
$p_i\in P$ s.t. ${p_i} \matcha {v}$.  We say that a set of patterns
$P$ \emph{subsumes} a pattern $p$ iff $\semg{p}\subseteq\semg{P}$.

Given a list of patterns $\PP=[p_1,\ldots,p_n]$, the
\emph{disambiguation} problem~\cite{ICFP-2008-Krauss} consists in
finding sets of patterns $P_1,\ldots,P_n$ such that for each $i\in
[1..n]$, $\semg{P_i}=\semg{p_i}\setminus\cup_{j=1}^{i-1}\semg{p_j}$.
Supposing the disambiguation problem can be solved, we have that for
any value $v$, $v\in\semg{P_i}$ iff $\PP\matcha_{i}{v}$. Consequently,
the definition of a function by a list of equations can be replaced by
an equivalent one consisting of a set of equations, {\ie} one where
the order of equations is not important.

The aforementioned properties of pattern matching can be also
expressed in terms of ground semantics.  Checking the {exhaustiveness}
of a list of patterns $\PP=[p_1,\ldots,p_n]$ consists in checking
whether for any value ${v}$ there exists an $i$ s.t.
$v\in\semg{p_i}\setminus\cup_{j=1}^{i-1}\semg{p_j}$. Checking if the
pattern $p_i$ is a {useless case} ({\wrt} $p_1,\ldots,p_{i-1}$)
consists in checking if there exists no value $v$
s.t. $v\in\semg{p_i}\setminus\cup_{j=1}^{i-1}\semg{p_j}$, {\ie}
checking whether $\{p_1,\ldots,p_{i-1}\}$ subsumes $p_i$ or not. For
the latter it is equivalent to check that
$\semg{p_i}\setminus\cup_{j=1}^{i-1}\semg{p_j}$ is empty and for the
former it is equivalent to check that
$\semg{x}\setminus\cup_{j=1}^{n}\semg{p_j}$ is empty.

We will come back to the use of disambiguation for generating
equivalent function definitions and detecting possible pattern
matching anomalies and for now we focus on  solving the
disambiguation problem.  To handle this problem we first define
\emph{extended patterns} as follows:
$$
\begin{array}{rclr}
  p & :=  &  \XX \mid c(p_1,\ldots,p_n) \mid {p_1}\plus{p_2}  \mid  {p_1}\minus{p_2}  \mid \bot 
  & \textit{with } c\in\CC
\end{array}
$$
Intuitively, a pattern ${p_1}\plus{p_2}$ matches any term matched by
one of its components. The relative \emph{complement} of $p_2$ {\wrt}
$p_1$, $p_1\minus p_2$, matches all terms matched by $p_1$ but those
matched by $p_2$. $\bot$ matches no term.  $\minus$ has a higher
priority than $\plus$.
If an extended pattern contains no $\minus$ it is called \emph{\addt}
and, if it contains no symbol $\bot$ is called \emph{pure}.

The pattern $p_1\plus p_2$ is linear if each of $p_1$ and $p_2$ is
linear; this corresponds to the fact that $p_1$ and $p_2$ represent
independent alternatives and thus, that their variables are unrelated
{\wrt}  pattern semantics.  For example, the terms $h(x)\plus g(x)$
and $h(x)\plus g(y)$ both represent all terms rooted by $h$ or $g$.
An extended pattern of the form $c(p_1,\ldots,p_n)$ is linear if each
$p_i$, $i\in [1..n]$, is linear and $\cap_{i=1}^n
\var{p_i}=\emptyset$. An extended pattern ${p_1}\minus{p_2}$ is linear if 
$p_1$, $p_2$ are linear and $\var{p_1}\cap\var{p_2}=\emptyset$.

In what follows we consider that all (tuples of) extended patterns are
linear and the set of all these patterns is denoted $\ETCX$.

The instance relation can be extended to take into account
extended patterns:
\[
\begin{array}{rcll}
 {p_1}\plus{p_2} & \matcha  &  v  & \text{iff }  {p_1} \matcha v ~\vee~ {p_2} \matcha  v  \\
 {p_1}\minus{p_2} & \matcha  &  v  & \text{iff }  {p_1} \matcha v ~\wedge~ {p_2} \nmatcha  v \\
 \bot & \nmatcha  &  v  
\end{array}
\]
with $p_1,p_2$ extended patterns and $v$ value.

The notion of ground semantics is also extended to take into account the
new constructions:
$$
\begin{array}{rcl}
\semg{p_1\plus p_2}&=&\semg{p_1}\cup\semg{p_2}
\\
\semg{p_1\minus p_2}&=&\semg{p_1}\setminus\semg{p_2}
\\
\semg{\bot}&=&\emptyset
\end{array}
$$
All notions apply as expected to tuples of extended patterns. We
generally use the term extended pattern to designate an extended
pattern or a tuple of extended patterns.

\begin{restatable}[Instance relation vs. ground semantics for extended
  patterns]{prop}{subsumptionVSsemanticsEXT}
  \label{th:subsumptionVSsemanticsEXT}
  Given an extended pattern ${p}$ and a value ${v}$, ${v}\in\semg{{p}}$ iff ${p}
  \matcha {v}$.
\end{restatable}

The disambiguation problem can be generalized to extended patterns:
given a list of extended patterns $[p_1,\ldots,p_n]$, the
disambiguation problem consists thus in finding sets of
\emph{constructor} patterns $P_1,\ldots,P_n$ such that for each $i\in
[1..n]$, $\semg{P_i}=\semg{p_i}\setminus\cup_{j=1}^{i-1}\semg{p_j}$.
When restricting to lists of constructor patterns we retrieve the
original disambiguation problem. By abuse of language, when we refer
to the disambiguation of a pattern we mean the disambiguation of the
list consisting only of this pattern; when the pattern is constructor
the disambiguation obviously results in the list containing only this
pattern.
Supposing this generalized disambiguation problem can be solved, the
definition of a function by a list of equations involving extended
patterns can be replaced by an equivalent one consisting of a set of
equations using only constructor patterns.

\section{Encoding extended patterns}
\label{se:encoding}

\begin{figure*}[!ht]
\framebox{\makebox[1\width]{
\begin{mathpar}
\begin{array}{l@{\hspace{5pt}}r@{\hspace{5pt}}c@{\hspace{5pt}}l@{\hspace{5pt}}l}
\multicolumn{2}{l}{\textbf{Remove empty sets:}}\\
\text{\small{(\rmAddL)}}&
  \bot \plus \meta{\vm}  & \raM & \meta{\vm} \\
\text{\small{(\rmAddR)}}&
  \meta{\vm} \plus \bot  & \raM & \meta{\vm} \\[+2pt]
\multicolumn{2}{l}{\textbf{Distribute sets:}}\\
\text{\small{(\rmEmptyF)}}&
h(\meta{v_1},\ldots,\bot{},\ldots,\meta{v_n})   & \raM & \bot \\[+2pt]
\text{\small{(\rmDistAdd)}}&
   h(\meta{v_1},\ldots,\meta{v_i}\plus\meta{w_i},\ldots,\meta{v_n}) & \raM &
   h(\meta{v_1},\ldots,\meta{v_i},\ldots,\meta{v_n}) \plus h(\meta{v_1},\ldots,\meta{w_i},\ldots,\meta{v_n}) \\
\multicolumn{2}{l}{\textbf{Simplify complements:}}\\
\text{\small{(\rmMinusMV)}}&
  \meta{\vm} \minus \varVar  & \raM & \bot \\
\text{\small{(\rmMinusMB)}}&
  \meta{\vm} \minus \bot   & \raM &  \meta{\vm}  \\
\text{\small{(\rmMinusMP)}}&
\meta{\wm} \minus (\meta{v_1} \plus \meta{v_2})   & \raM &  (\meta{\wm} \minus \meta{v_1}) \minus \meta{v_2} \\
\text{\small{(\rmMinusVM)}}&
  \varVar \minus g(\meta{t_1},\ldots,\meta{t_n})  & \raM & 
  \sum\limits_{c\in\CC} c(\varCt_1,\ldots,\varCt_{m}) \minus g(\meta{t_1},\ldots,\meta{t_n}) &
  {\textit{with  } m=arity(c)} \\
\text{\small{(\rmMinusBM)}}&
  \bot \minus f(\meta{v_1},\ldots,\meta{v_n})  & \raM & \bot \\
\text{\small{(\rmMinusPM)}}&
  (\meta{\vm} \plus \meta{\wm}) \minus f(\meta{v_1},\ldots,\meta{v_n})   & \raM &  (\meta{\vm} \minus f(\meta{v_1},\ldots,\meta{v_n})) \plus (\meta{\wm} \minus f(\meta{v_1},\ldots,\meta{v_n})) \\
\text{\small{(\rmMinusFF)}}&
  f(\meta{v_1},\ldots,\meta{v_n}) \minus f(\meta{t_1},\ldots,\meta{t_n})  & \raM &
  f(\meta{v_1}\minus \meta{t_1},\ldots,\meta{v_n}) \plus \cdots \plus f(\meta{v_1},\ldots,\meta{v_n}\minus \meta{t_n})   \\
\text{\small{(\rmMinusFG)}}&
  f(\meta{v_1},\ldots,\meta{v_n}) \minus g(\meta{w_1},\ldots,\meta{w_n})  & \raM &
  f(\meta{v_1},\ldots,\meta{v_n}) 
  & {\textit{with  } f\not=g}\\
\end{array}
\end{mathpar}
}}
\caption{$\RrC$: reduce extended patterns to {\addt} terms.
  $\meta{\vm},\meta{v_1},\ldots,\meta{v_n}$,
  $\meta{\wm},\meta{w_1},\ldots,\meta{w_n}$ range over {\addt}
  patterns, $\meta{t_1},\ldots,\meta{t_n}$ range over pure {\addt}
  patterns, $\varVar$ ranges over pattern variables. $f,g$ expand to
  all the symbols in $\CC\cup\CCt$, $h$ expands to all
  symbols in $\CC^{n>0}\cup\CCt$.
\label{fig:elimComplement}
}
\end{figure*}

To solve the disambiguation problem we propose a method for
transforming any extended pattern $p$ and, in particular, any
complement pattern, into an equivalent pure {\addt} pattern
$p_1\plus\cdots\plus p_n$ and thus obtain the set of constructor
patterns $\{p_1,\cdots,p_n\}$ having the same semantics as the
original one; if $p$ is transformed into $\bot$ then it is useless.
This transformation is accomplished using the rewriting system $\RrC$
presented in Figure~\ref{fig:elimComplement}.  For simplicity, this
rewriting system is presented schematically using rules 
which abstract over the symbols of the signature.  We use overlined
symbols, like $\meta{\tm}$,
$\meta{\vm},\meta{\wm}$, to denote
the variables of the TRS and $\varCt$ to denote (freshly generated)
pattern-level variables.
We will show that each intermediate step and consequently the overall
transformation is sound and complete {\wrt} the ground semantics.

Rules $\rmAddL$ and $\rmAddR$ express the fact that the empty ground
semantics of $\bot$ is neutral for the union.
Rule $\rmEmptyF$ indicates that the semantics of a pattern containing
a sub-term with an empty ground semantics is itself empty.
Similarly, if the semantics of a sub-term can be expressed as the
union of two sets then the semantics of the overall term is obtained
by distributing these sets over the corresponding constructors; this
behaviour is reflected by the rule $\rmDistAdd$.
Note that $\rmEmptyF$ and $\rmDistAdd$ are rule schemes representing
as many rules as constructors of strictly positive arity in the
signature and tuple symbols in $\CCt$.

The remaining rules describe the behaviour of complements and
generally correspond to set theory laws over the ground semantics of
the involved patterns.  The difference between the ground semantics of
any pattern and the ground semantics of a variable, which corresponds
to the set of all ground constructor patterns for the signature, is
the empty set; rule $\rmMinusMV$ encodes this behaviour.  When
subtracting the empty set, the argument remains unchanged (rule
$\rmMinusMB$).  Subtracting the union of several sets consists in
subtracting successively all sets (rule $\rmMinusMP$).
The semantics of a variable is the set of all ground constructor patterns,
set which  can be also obtained by considering for each  constructor
in the signature the set of all terms having this symbol at the root
position and taking the union of all these sets (rule $\rmMinusVM$).
We should emphasize that $\varVar$ is a  variable ranging
over pattern variables at the object level and that $\varCt_i$ are fresh
pattern variables seen as constants at the TRS level ({\ie}
$\varVar$ matches any $\varCt_i$).
Similarly to rules $\rmMinusMV-\rmMinusMP$, rules $\rmMinusBM$ and $\rmMinusPM$
correspond to their counterparts from set theory. 
Rule $\rmMinusFF$ corresponds to the set difference of cartesian
products; the case when the head symbol is a constant $c$ corresponds
to the rule $c \minus c \raM \bot$.
Rule $\rmMinusFG$  corresponds just to the special
case where complemented sets are disjoint.

It is worth noticing that the rule schemes $\rmMinusVM-\rmMinusFG$
expand to all the possible rules obtained by replacing $f,g$ with all
the constructors in the original signature and all tuple symbols.
Note also that the variables in the rewrite rules range over (pure)
additive patterns which correspond implicitly to a call-by-vallue
reduction strategy.

\begin{example}
  \label{ex:appComp}
  Let us consider the signature $\Sigma$ with $\CC=\{a,b,f\}$ and
  $\ar(a)=\ar(b)=0$, $\ar(f)=2$.
  The pattern $f(x,y)\minus f(z,a)$ corresponds to all patterns rooted
  by $f$ but those of the form $f(z,a)$. According to rule
  $\rmMinusFF$ this corresponds to taking all patterns rooted by $f$
  which are not discarded by the first argument of $f(z,a)$, {\ie}
the pattern
  $f(x\minus z,y)$, or by its second argument, {\ie} 
the pattern
  $f(x,y\minus a)$. We obtain thus the pattern $f(x\minus
  z,y)\plus{f(x,y\minus a)}$ which reduces, using rule $\rmMinusMV$
  and the propagation and elimination of $\bot$ to $f(x,y\minus
  a)$. Using rule $\rmMinusVM$ we obtain $f(x,(a\plus b\plus
  f(y_1,y_2))\minus a)$ which reduces eventually to $f(x,b\plus
  f(y_1,y_2))$. We can then apply $\rmDistAdd$ to obtain the term
  $f(x,b) \plus f(x,f(y_1,y_2))$ which is irreducible.

  The rewrite rules apply also on tuples of patterns and 
  $\tup{x,y}\minus \tup{z,a}$ reduces using the same rules as above to
  $\tup{x,b} \plus \tup{x,f(y_1,y_2)}$.
  Similarly $\tup{x,y} \minus \tup{b,a}$ reduces to $\tup{a\plus
    f(x_1,x_2),y}\plus \tup{x,b\plus f(x_1,x_2)}$ and then to the
  irreducible term $(\tup{a,y} \plus \tup{f(x_1,x_2),y}) \plus
  (\tup{x,b} \plus \tup{x,f(x_1,x_2)})$.
\end{example}

\begin{restatable}[Convergence]{lemm}{convergence}
\label{lemma:convergence}
  The rewriting system $\RrC$ is confluent and terminating.  The normal
  form of an extended pattern {\wrt} to $\RrC$ is either $\bot{}$ or a
  sum of (tuples of) constructor patterns, {\ie} a pure {\addt} term $t$
  such that if $t(\omega)={\plus}$ for a given $\omega$ then, for all
  $\prefix{\omega'}{\omega}$, $t(\omega')=\plus$.
\end{restatable}

Note that since the rewrite rules introduce only fresh pattern
variables (rule $\rmMinusVM$) and duplicate terms only through $\plus$
(rules $\rmMinusPM$, $\rmMinusFF$ and $\rmDistAdd$), a linear term is
always rewritten to a linear term and thus, the normal form of a
linear term is linear as well.

As intuitively explained above, the reduction preserves the ground
semantics of linear terms:

\begin{restatable}[Complement semantics preservation]{prop}{comEq}
\label{prop:comEq}
  For any extended patterns $p,p'$, 
  if $p\evalM_{\RrC}p'$ then $\semg{p}=\semg{p'}$.
\end{restatable}

Checking whether a (extended) pattern $p$ is useless {\wrt} a set of
patterns $\{p_1,\ldots,p_n\}$ can be done by simply verifying that the
pattern $p\minus (p_1\plus\cdots\plus p_n)$ is reduced by $\RrC$ to
$\bot$, meaning that this pattern has an empty semantics:
%
\begin{restatable}[Subsumption]{prop}{subsumption}
\label{th:subsumption}
  Given the patterns $p,p_1,\ldots,p_n$,
  $p \textrm{ is subsumed by } \{ p_1,\ldots,p_n \}$
  iff 
  $p\minus(p_1\plus\cdots\plus p_n)\downarrow_{\RrC}=\bot$.
\end{restatable}

\begin{example}
  \label{ex:uselessPattern}
  We consider the signature in Example~\ref{ex:appComp} and the list
  of patterns
  $[ \tup{b,y}\otrsep \tup{a,b}\otrsep$ $\tup{f(x,y),z}\otrsep$ $\tup{x,b} ]$.
  To check if the last pattern in the list is useless it is
  enough to verify whether the pattern $\tup{x,b} \minus (\tup{b,y}\plus
  \tup{a,b}\plus \tup{f(x,y),z})$ reduces to $\bot$ or not.  
  The pattern $\tup{x,b} \minus \tup{b,y}$ reduces to $\tup{a\plus f(x_1,x_2),b}$
  and when we further subtract $\tup{a,b}$ we obtain
  $\tup{f(x_1,x_2),b}$. Finally, $\tup{f(x_1,x_2),b}\minus \tup{f(x,y),z}$
  reduces to $\bot$ and we can thus conclude that the pattern
  $\tup{x,b}$ is useless {\wrt} the previous patterns in the list.

  One may want to check the exhaustiveness of the list of patterns
  $[ \tup{b,y}\otrsep \tup{a,b}\otrsep$ $\tup{f(x,y),z} ]$. 
  Since the pattern
  $\tup{x,y}\minus (\tup{b,y}\plus \tup{a,b}\plus \tup{f(x,y),z})$ reduces to
  $\tup{a,a}$ we can conclude that the property doesn't hold. We can then
  check similarly that exhaustiveness holds for the list of patterns
  $[ \tup{b,y}\otrsep \tup{a,b}\otrsep$ $\tup{f(x,y),z} \otrsep \tup{a,a} ]$.
\end{example}

With the transformation realized by $\RrC$ an extended pattern is
transformed into an equivalent {\addt} one with $\bot{}$ potentially
present only at the root position and with all sums pushed at the top
level ({\ie} until there is no $\plus$ symbol below any other
symbol). More precisely, if we abstract over the way $\plus$
associates, any extended pattern $p$ is normalized {\wrt} $\RrC$ into
$\bot{}$ or into a sum of (tuples of) constructor patterns
$p_1\plus\cdots\plus p_n$ having the same semantics as $p$.  Since the
semantics of this latter pattern is exactly the same as the semantics
of the set $P=\{p_1,\ldots, p_n\}$, the above transformations can be
used to solve the disambiguation problem. If the result of the
reduction of $p$ is $\bot$ then $P=\emptyset$ and in this case the
pattern is useless.
 
\begin{example}
  \label{ex:simple}
  Let us consider the signature from Example~\ref{ex:appComp} and the
  list of patterns $[f(x,y),f(z,a)]$.  As we have seen, the pattern
  $f(x,y)\minus f(z,a)$ reduces {\wrt} $\RrC$ to $f(x,b)\plus
  f(x,f(y_1,y_2))$ and thus, the original list of patterns is
  disambiguated into the sets of patterns $\{f(x,y)\}$ and
  $\{f(x,b),f(x,f(y_1,y_2))\}$.
\end{example}

The above transformation can be also used as a generic compilation
method for the so-called {\antit}s~\cite{jsc2010}, {\ie} a method for
transforming an {\antit} into an extended pattern and eventually into
a set of constructor patterns having the same semantics as the
original {\antit}. An {\antit} is a linear term in $\mathcal{T(C\cup\ap,X)}$\footnote{In their most general form {\antit}s are not
necessarily linear.} and, intuitively, the semantics of an {\antit}
represents the complement of its semantics with respect to
$\TC$. Formally~\cite{jsc2010},
$\semg{\rmp{t}{\omega}{\ap{t'}}}=\semg{\rmp{t}{\omega}{z}}\setminus\semg{\rmp{t}{\omega}{t'}}$
where $z$ is a fresh variable and for all $\prefix{\omega'}{\omega}$,
$t(\omega')\neq\ap$.  For example, the complement of a variable $\ap
x$ denotes $\TC\setminus\semg{x} = \TC\setminus\TC =
\emptyset$. Similarly, $\ap g(x)$ denotes $\TC \setminus \{g(t)\mid
t\in\TC\}$, and $f(\ap a,x)$ denotes $\{f(v,u)\mid v,u\in\TC\}
\setminus \{f(a,u)\mid u\in\TC\}$.

The compilation is simply realized by replacing all {\antit}s by their
absolute complement; this replacement can be expressed by a single
rewrite rule $\RrAP=\{\ap \meta{\tm}\raM \varCt\minus
\meta{\tm}\}$ where $\meta{t}$ is a variable ranging over {\antit}s
and $\varCt$ corresponds to a \emph{fresh} pattern-level variable
({\ie} a variable of the pattern being transformed).

\begin{example}
  \label{ex:antip}
  We have $f(x,\ap a)\evalM_{\RrAP} f(x,y\minus a)$ where $y$
  is a fresh variable; this pattern reduces {\wrt} $\RrC$ to
  $f(x,b)\plus f(x,f(y_1,y_2))$.
  Similarly $\ap f(x,\ap a)\evalM_{\RrAP} z\minus f(x,y\minus
  a)$ with $y,z$ fresh variables and the latter pattern reduces to
  $a\plus b\plus f(x,a)$.
\end{example}

$\RrAP$ is clearly convergent and the normal form of any {\antit} is
an extended term containing no $\ap$ symbol.
Since the reduction introduces only fresh variables and does not
duplicate terms, the normal form of a linear term is linear as well.
Moreover, the reduction preserves the ground semantics:

\begin{restatable}[Anti-pattern semantics preservation]{prop}{apEq}
\label{prop:apEq}
  For any  {\antit}s $p,p' \in\ETCXap$,
  if $p\evalM_{\RrAP}p'$ then, $\semg{p}=\semg{p'}$.
\end{restatable}

In the rest of this paper we will thus consider
that an anti-pattern is just syntactic sugar for the corresponding
extended pattern obtained by replacing all its sub-terms of the form
$\ap q$ by $z\minus q$ with $z$ a fresh variable.

\section{Elimination of redundant patterns}
\label{se:elimination}
%
%
We have so far a method for transforming an extended (anti-)pattern
$p$ into a set of constructor patterns~$P$.  The set $P$ is not
necessarily canonical and can contain, for example, duplicate or
redundant patterns, {\ie} patterns useless {\wrt} the other patterns
in $P$.

\begin{example}\label{ex:redundantRule}
  The pattern $f(x,!a)\minus f(b,a)$ which corresponds to
  $f(x,y\minus a)\minus f(b,a)$ is reduced by
  $\RrC$ to $f(a\plus f(x_1,x_2),b\plus f(y_1,y_2)) \plus f(x,b\plus
  f(y_1,y_2))$ and finally to the pure additive pattern $f(a,b) \plus
  f(f(x_1,x_2),b) \plus f(a,f(y_1,y_2)) \plus f(f(x_1,x_2),f(y_1,y_2))
  \plus f(x,b) \plus f(x,f(y_1,y_2))$.  The disambiguation of the
  initial pattern results thus in the set
  $\{f(a,b)\otrsep\;\;
  f(f(x_1,x_2),b)\otrsep\;\; f(a,f(y_1,y_2))\otrsep\;\;
  f(f(x_1,x_2),$ \linebreak $f(y_1,y_2))\otrsep\; f(x,b)\otrsep\; f(x,f(y_1,y_2))\}$
  which is clearly equivalent to the set $\{f(x,b) \otrsep\;
  f(x,f(y_1,y_2))\}$ since all the patterns of the former are subsumed
  by the patterns of the latter.
\end{example}

The simplification consisting in eliminating patterns subsumed by
other patterns is obvious and this is one of the optimizations
proposed in Section~\ref{se:experiments}.  There are some other cases
where a pattern is subsumed not by a single pattern but by several
ones.
The objective is to find, for each set $P$ of constructor patterns
resulting from the transformation of an extended pattern a smallest
subset $P'\subseteq P$ such that $P'$ has the same 
semantics as $P$.
In particular, a pattern $p_k$ from $P=\{p_1,\ldots,p_n\}$ can be
removed from $P$ without changing its semantics if
$\semg{p_k}\subseteq\bigcup_{j\neq k}\semg{p_j}$.
By exploring all possible removals we can find the smallest
subset~$P$.

\begin{example}
  \label{ex:elimPattern}
  We consider the signature from Example~\ref{ex:appComp} enriched
  with the 
  constructor $g$ with $\ar(g)=1$ and the set of constructor patterns $\{
    f(g(b),f(x,b)) \otrsep$
    $f(g(b),f(b,y)) \otrsep$
    $f(g(x),f(a,b)) \otrsep$
    $f(x,f(f(z_1,z_2),y)) \otrsep$
    $f(x,f(g(z),y)) 
  \}$.
  This time none of the patterns is subsumed directly by another one
  but the first one is subsumed by the set consisting of the four
  other patterns.  To convince ourselves we can consider instances of
  this pattern with~$x$ replaced respectively by $a, b, g(z)$ and
  $f(z_1,z_2)$ ({\ie} all the constructors of the signature) and check
  that each of these instances is subsumed by one of the other
  patterns.
\end{example}

We have seen that we can identify redundant patterns in a set $P$
(Proposition~\ref{th:subsumption}) and thus we can subsequently remove
them in order to obtain a \emph{valid subset} $P'\subseteq P$ with
equivalent semantics, $\semg{P'}=\semg{P}$.
Given a set of patterns we can remove all redundant patterns one by
one till the obtained set contains no such pattern but, depending on
the pattern we have chosen to eliminate at some point, we can
nevertheless get different valid subsets and some of them do not
necessarily lead to a minimal one.

\begin{figure}[t!]
  \begin{framed}
{\flushleft
$\minimum(P) = \minimum'(P,\emptyset)$\\
$\minimum'(\emptyset, kernel) = kernel$\\
$\minimum'(\{q\}\cup P, kernel) =$ \textsf{\textbf{if}} $q$ \textrm{is subsumed by} $P\cup kernel$ \textsf{\textbf{then}}\\
$\;\;\;\;\;\;\;\;\;\;\;\;\;\;\;\;\;\;\;\;\;\;\;\;\;\;\;\;\;\;\;\;\;\;\;\;\;\;\;\;\;\;\;\;\;\;\;$
$smallest\_set( \minimum'(P, \{q\}\cup kernel), \minimum'(P,kernel) )$\\
$\;\;\;\;\;\;\;\;\;\;\;\;\;\;\;\;\;\;\;\;\;\;\;\;\;\;\;\;\;\;\;\;\;\;\;\;\;\;\;\;\;\;\;$
\textsf{\textbf{else}}\\
$\;\;\;\;\;\;\;\;\;\;\;\;\;\;\;\;\;\;\;\;\;\;\;\;\;\;\;\;\;\;\;\;\;\;\;\;\;\;\;\;\;\;\;\;\;\;\;$
$\minimum'(P, \{q\}\cup kernel)$\\
}
\end{framed}
\caption{$\minimum(P)$ computes a minimal valid subset of~$P$.
$smallest\_set(P,P')$ returns $P$ if $\card{P}<\card{P'}$, $P'$
  otherwise.
\label{fig:minimize}
}
\end{figure}

Computing the smallest valid subset can be done by enumerating the
powerset of~$P$ and taking its smallest element $P'$ which is a valid
subset of $P$.  Figure~\ref{fig:minimize} presents a more efficient
algorithm where the search space is reduced: $P'$ is searched only
among the subsets of~$P$ which contain the initial $kernel$ of $P$,
{\ie} the set $\{p\mid p\in P,p\text{ is not subsumed by
}P\minus\{p\}\}$.  The algorithm still explores all the possible valid
subsets and eventually returns the minimal one:
%
\begin{restatable}[Minimal subset]{prop}{minimize}
  \label{prop:minimize}
  Given a set of constructor patterns $P$, the algorithm given in
  Figure~\ref{fig:minimize} computes the smallest valid subset
  $P'\subseteq P$.
\end{restatable}

\section{Function encoding} 
\label{se:caseEncode}
%
We have focused so far on the matching mechanism behind function
definitions using case expressions and we have eluded so far the
potential problems related to the evaluation of such functions.

If we consider, for example, a function $\foo$ defined by the list of
rules 
\[\arraycolsep=1.4pt
\begin{array}{ll}
[
&
\foo(z,a) \ra z 
\otrsep
\\
&
\foo(x,y) \ra y
~~]
\end{array}
\]
we can proceed to the disambiguation of its patterns which results in
the sets of patterns \{$\tup{z,a}$\} and $\{\tup{x,b}\otrsep
\tup{x,f(y_1,y_2)}\}$ as shown in the examples in the previous section.
Consequently, if we replace naively the initial patterns with the ones
obtained by disambiguation then  the following set of  corresponding
rules is obtained
\[\arraycolsep=1.4pt
\begin{array}{ll}
\{
&
\foo(z,a) \ra z 
\otrsep
\\
&
\foo(x,b) \ra y
\otrsep
\\
&
\foo(x,f(y_1,y_2)) \ra y
~~\}
\end{array}
\]
One can easily see that the two last rules are not well-defined and in
what follows we extend the transformations proposed in the previous
sections to tackle such situations.

\subsection{As-patterns and their encoding}
\label{se:aspatterns}

\begin{figure*}[!t]
\framebox{\makebox[1\width]{
\begin{mathpar}
\begin{array}{l@{\hspace{5pt}}r@{\hspace{5pt}}c@{\hspace{5pt}}l@{\hspace{5pt}}l}
\multicolumn{2}{l}{\textbf{Remove empty sets:}}\\
\text{\small{(\rmAddL)}}&
  \bot \plus \meta{\vm}  & \raM & \meta{\vm} \\
\text{\small{(\rmAddR)}}&
  \meta{\vm} \plus \bot  & \raM & \meta{\vm} \\[+2pt]
\multicolumn{2}{l}{\textbf{Distribute sets:}}\\
\text{\small{(\rmEmptyF)}}&
h(\meta{v_1},\ldots,\bot{},\ldots,\meta{v_n})   & \raM & \bot \\[+2pt]
\text{\small{(\rmEmptyA)}}&
  \varVar \at \bot{}   & \raM & \bot \\[+2pt]
\text{\small{(\rmDistAdd)}}&
   h(\meta{v_1},\ldots,\meta{v_i}\plus\meta{w_i},\ldots,\meta{v_n}) & \raM &
   h(\meta{v_1},\ldots,\meta{v_i},\ldots,\meta{v_n}) \plus h(\meta{v_1},\ldots,\meta{w_i},\ldots,\meta{v_n}) \\[+2pt]
\text{\small{(\rmDistAt)}}&
\varVar\at(\meta{v_1}\plus \meta{v_2}) & \raM & \varVar\at\meta{v_1} \plus \varVar\at\meta{v_2} \\[+2pt]
\multicolumn{2}{l}{\textbf{Simplify complements:}}\\[+2pt]
\text{\small{(\rmMinusMV)}}&
  \meta{\vm} \minus \varVar  & \raM & \bot \\[+2pt]
\text{\small{(\rmMinusMB)}}&
  \meta{\vm} \minus \bot   & \raM &  \meta{\vm}  \\[+2pt]
\text{\small{(\rmMinusMP)}}&
\meta{\wm} \minus (\meta{v_1} \plus \meta{v_2})   & \raM &  (\meta{\wm} \minus \meta{v_1}) \minus \meta{v_2} \\[+2pt]
\text{\small{(\rmMinusVMp)}}&
  \varVar \minus g(\meta{t_1},\ldots,\meta{t_n})  & \raM & 
  \varVar \at (
  \sum\limits_{c\in\CC} c(\varCt_1,\ldots,\varCt_{m}) \minus g(\meta{t_1},\ldots,\meta{t_n}) )&
  {\textit{with  } m=arity(c)} \\[+2pt]
\text{\small{(\rmMinusBM)}}&
  \bot \minus f(\meta{v_1},\ldots,\meta{v_n})  & \raM & \bot \\[+2pt]
\text{\small{(\rmMinusPM)}}&
  (\meta{\vm} \plus \meta{\wm}) \minus f(\meta{v_1},\ldots,\meta{v_n})   & \raM &  (\meta{\vm} \minus f(\meta{v_1},\ldots,\meta{v_n})) \plus (\meta{\wm} \minus f(\meta{v_1},\ldots,\meta{v_n})) \\[+2pt]
\text{\small{(\rmMinusFF)}}&
  f(\meta{v_1},\ldots,\meta{v_n}) \minus f(\meta{t_1},\ldots,\meta{t_n})  & \raM &
  f(\meta{v_1}\minus \meta{t_1},\ldots,\meta{v_n}) \plus \cdots \plus f(\meta{v_1},\ldots,\meta{v_n}\minus \meta{t_n})   \\[+2pt]
\text{\small{(\rmMinusFG)}}&
  f(\meta{v_1},\ldots,\meta{v_n}) \minus g(\meta{w_1},\ldots,\meta{w_n})  & \raM &
  f(\meta{v_1},\ldots,\meta{v_n}) 
  & {\textit{with  } f\not=g}\\[+2pt]
\text{\small{(\rmMinusAM)}}&
  \varVar \at \meta{\vm} \minus \meta{\wm}  & \raM &    \varVar \at (\meta{\vm} \minus \meta{\wm}) \\[+2pt]
\text{\small{(\rmMinusMA)}}&
  \meta{\vm} \minus \varVar \at \meta{\wm}  & \raM &    \meta{\vm} \minus \meta{\wm} \\[+2pt]
\end{array}
\end{mathpar}
}}
\caption{$\RrCat$: reduce (as-)extended patterns to {\addt} terms.
  $\meta{\vm},\meta{v_1},\ldots,\meta{v_n}$,
  $\meta{\wm},\meta{w_1},\ldots,\meta{w_n}$ range over {\addt}
  patterns, $\meta{t_1},\ldots,\meta{t_n}$ range over pure {\addt}
  patterns, $\varVar$ ranges over pattern variables. $f,g$ expand to
  all the symbols in $\CC\cup\CCt$, $h$ expands to all
  symbols in $\CC^{n>0}\cup\CCt$.
\label{fig:elimComplementAs}
}
\end{figure*}

We first consider a new construct for \emph{extended
patterns} which are now defined as follows:
$$
\begin{array}{rclr}
  p & :=  &  \XX \mid f(p_1,\ldots,p_n) \mid {p_1}\plus{p_2}  \mid  {p_1}\minus{p_2}  \mid \bot 
  \mid   q\at{p}
\end{array}
$$
with $f\in\CC, q\in\TCX$.

All patterns $q\at{p}$, called \emph{as-patterns}, are, as all the
other extended patterns, linear {\ie} $p$, $q$ are linear and
$\var{q}\cap\var{p}=\emptyset$.
As we will see, $\at$ is a convenient way to alias terms and use the
variable name in the right-hand side of the corresponding rewrite
rule.  In fact, all the aliases used explicitly in the left-hand sides
of the {\er}s (defined formally in Section~\ref{se:reduction}) are
of the form $x\at p$; the general form $q@p$ is used only in the
matching process and in this case $q\in\TF$.
$\at$ has a higher priority than $\minus$
which has a higher priority than $\plus$.
From now on, unless stated explicitly, extended patterns are
considered to include as-patterns.

The notion of ground semantics is extended accordingly for as-patterns:
$\semg{{q\at{p}}}=\semg{q}\cap\semg{p}$.
Notice that the variable $x$ aliasing the pattern $p$ in $x\at p$
has no impact on the semantics of the term: $\semg{x\at p}=\semg{p}$.

To transform any extended (as-)pattern into a pure {\addt} pattern we
use the rewriting system $\RrCat$ described in
Figure~\ref{fig:elimComplementAs}; it consists of the the rules of
$\RrC$ with the rule $\rmMinusVM$ slightly modified together with a set
of specific rules used to handle the as-patterns.

The new rules $\rmEmptyA$ and $\rmDistAt$ specify respectively that
aliasing a $\bot$ is useless and that aliasing a sum comes to aliasing
all its patterns. Rules $\rmMinusAM$ and $\rmMinusMA$ indicate that
the alias of a complement pattern $p\minus{q}$ concerns only the
pattern $p$. The modified rule $\rmMinusVMp$ guarantees that the
variables of a complement pattern are not lost in the transformation
and, as we will see in the next sections, prevent ill-formed rules as
those presented at the beginning of the section.

\begin{example}
  \label{ex:appCompAs}
  We consider the signature in Example~\ref{ex:appComp}.
  The pattern $f(x,y)\minus f(z,a)$ reduces {\wrt} $\RrCat$ as it had
  {\wrt} $\RrC$ but because of the new rule $\rmMinusVMp$ we obtain
  $f(x,y\at (b\plus f(y_1,y_2)))$. This latter term is eventually
  reduced using rules $\rmDistAt$ and $\rmDistAdd$ to $f(x,y\at b)\plus
  f(x,y\at f(y_1,y_2))$.
\end{example}

$\RrCat$ is convergent and 
the normal form of a term in $\ETCX$ is similar to that that obtained
with $\RrC$ but with some of its subterms potentially aliased with the
$\at$ construct.

\begin{restatable}[Convergence]{lemm}{convergenceAt}
  The rewriting system $\RrCat$ is confluent and terminating.  Given
  an extended pattern $t$ 
  the normal form of $t$ {\wrt} to $\RrCat$ is either $\bot{}$ or a
  sum of (tuples of) constructor patterns potentially aliased, {\ie} a
  pure {\addt} term $t$ such that if $t(\omega)={\plus}$ for a given
  $\omega$ then, for all $\prefix{\omega'}{\omega}$,
  $t(\omega')=\plus$.
\end{restatable}

For the same reasons as before, a linear term is always rewritten to a
linear term and thus, the normal form of a linear term is linear as
well.  Once again, each intermediate step and consequently the overall
transformation is sound and complete {\wrt} the ground semantics.

\begin{restatable}[Complement semantics preservation]{prop}{comEqAs}
\label{prop:comEqAs}
  For any extended pattern $p,p'$ 
  if $p\evalM_{\RrCat}p'$ then, $\semg{p}=\semg{p'}$.
\end{restatable}

\subsection{Matchable and free variables}
\label{se:matchable}
%
Given a constructor pattern $p$ and a value $v$, if $v\in\semg{p}$
then there exists a substitution $\sigma$ with $\dom{\sigma}=\var{p}$
{\st} $v\in\semg{\sigma(p)}$, or equivalently $v=\semg{\sigma(p)}$. When $p$
is an extended pattern some of its variables are not significant for
the matching, {\ie} if $v\in\semg{p}$ then there exists a substitution
$\sigma$ {\st} $v\in\semg{\tau(\sigma(p))}$ for all substitution
$\tau$ with $\dom{\tau}=\var{p}\setminus\dom{\sigma}$.
For example, given the pattern 
$f(x,z\minus g(y))$, the value $f(a,b)$ belongs to the semantics of
any instance of $f(a,b\minus g(y))$.

The set $\mvar p$ of \emph{matchable variables} of a pattern $p$ is
defined as follows:

\[
\begin{array}{l@{\hspace{15pt}}l}
  \mvar{x} = \{x\}, \forall x\in\XX  
  \\
  \mvar{f(p_1,\ldots,p_n)} = \mvar{p_1}\cup\ldots\cup\mvar{p_n}, \forall f \in \CC^n
  \\
  \mvar{p_1\plus p_2} = \mvar{p_1}\cap\mvar{p_2} 
  \\
  \mvar{p_1\minus p_2} =  \mvar{p_1} 
  \\
  \mvar{\bot} = \XX 
  \\
  \mvar{q\at p} =  \mvar{q}\cup\mvar{p} 
\end{array}
\]
The variables of $p$ which are not matchable are \emph{free}: 
\[
\fvar{p}=\var{p}\setminus\mvar{p}
\]

Note that the definition of linearity we have used for (complement)
patterns guarantees that matchable and free variables have different
names. Consequently, we have $\sigma(p_1\minus
p_2)=\sigma(p_1)\minus{p_2}$ for all $\sigma$ such that
$\dom{\sigma}=\mvar{p_1\minus p_2}$.

The encoding rules preserve not only the semantics but also the set of
matchable variables of the initial pattern; this is important when
transforming the {\er}s introduced in the next section.

\begin{restatable}[Variables preservation]{prop}{mvarPreserve}
\label{prop:mvarPreserve}
  For any extended patterns $p,p'$ such that
  $p\evalM_{\RrCat}p'$
  we have $\mvar{p}\subseteq\mvar{p'}$ and $\fvar{p'}\subseteq\fvar{p}$.
\end{restatable}

We consider for convenience that the set of matchable variables of
$\bot$ is the set of all variables; a more natural definition
considering the empty set would have required a more complicated
statement for the above proposition dealing explicitly with the rules
whose right-hand side is $\bot$. As explained in the next section this
choice has no impact on the proposed formalism.

The above property was not verified by the rule $\rmMinusVM$ of $\RrC$
and it was intuitively the origin of the ill-formed rules
presented at the beginning of the section.
An immediate consequence of  this property is that for any
patterns $p,p_1,p_2,\ldots,p_n$ such that
$p\downarrow_{\RrCat}=p_1\plus p_2\plus\ldots\plus p_n$ we have
$\mvar{p}\subseteq\mvar{p_i}$ for all $i\in[1,\ldots,n]$.

\subsection{Encoding sets and lists of {\er}s}
\label{se:reduction}
%

An \emph{\erw}, or simply \emph{\er}, over a signature $\Sigma$ is a
pair $(l,r)$ (also denoted $l \ra r$) with $l=\foo(\ovr{p})$,
$\foo\in\DD$, $\ovr{p}$ a tuple of extended patterns, and $r\in\TFX$,
such that $\var r\subseteq \mvar{\ovr{p}}$.
A set of {\er}s $\EE=\{\foo(\ovr{p_1})\ra
r_1,\ldots,\foo(\ovr{p_n})\ra r_n\}$ induces a relation
$\eval_{\EE}$ over $\TF$ such that $t \eval_{\EE}
t'$ iff there exist $i\in[1,\ldots,n]$, $\omega\in\PPos(t)$ and a
substitution $\sigma$ such that $\stt t \omega = \foo(\ovr{v})$,
$\dom{\sigma}=\mvar{\ovr{p_i}}$, $\ovr{v} \in
\semg{\sigma(\ovr{p_i})}$ and $t'=\rmp t \omega {\sigma(r_i)}$.
A list of {\er}s $\LL=[\foo(\ovr{p_1})\ra
r_1,\ldots,\foo(\ovr{p_n})\ra r_n]$ induces a relation
$\eval_{\LL}$ over $\TF$ such that $t \eval_{\LL}
t'$ iff there exist $i\in[1,\ldots,n]$, $\omega\in\PPos(t)$ and a
substitution $\sigma$ such that $\stt t \omega = \foo(\ovr{v})$,
$\dom{\sigma}=\mvar{\ovr{p_i}}$, $\ovr{v} \in
\semg{\sigma(\ovr{p_i})}$, ${p_j} \nmatcha {v}, \forall j<i$ and
$t'=\rmp t \omega {\sigma(r_i)}$.  We may write $\foo(\ovr{v})
\evaleps{i}_{\LL} t'$ to indicate that the $i$-th {\er} has been
used in the reduction.

When restricting to constructor patterns, all variables of a pattern
$p$ are matchable and thus, for all $\sigma$ such that
$\dom{\sigma}=\mvar{p}$, $\sigma(p)$ is ground and for all value $v$
we have ${v}\in\semg{\sigma(p)}$ iff ${v}=\sigma(p)$. Consequently,
sets and lists of {\er}s whose left-hand sides contain only
constructor patterns are nothing else but CTRSs and respectively
ordered CTRSs. 
We have used a different syntax for {\er}s and general rules in order
to help the reader identify the rules used in function definitions from
those used to transform the extended patterns.

In the rest of this section we show how these lists of {\er}s, which
can be seen as extended ordered CTRSs, can be compiled towards plain
ordered CTRS and eventually plain (order independent) CTRS.

Note that since we considered that the set of matchable variables of
$\bot$ is the set of all variables, we can have {\er}s with a
left-hand side $\bot$ and an arbitrary right-hand. Such {\er}s are
of no practical use since they apply on no term. For the rest of this
paper we restrict thus to {\er}s containing no $\bot$.

\begin{example}
\label{ex:equationLists}
We consider 
a signature 
with
 $\DD=\{\foo\},$
$\CC=\{a,b,f\}$, $\ar(\foo)=2$, $\ar(a)=\ar(b)=0$, $\ar(f)=2$  and
the list of {\er}s
$\LL={[ \foo(x, y\at\ap a) \ra y, \foo(a \plus b,y) \ra y, \foo(f(x,y),z) \ra x]}$.
We have the reductions
 $\foo(a,b) \evaleps{1}_{\LL} b$,
 $\foo(a,a) \evaleps{2}_{\LL} a$,
 $\foo(b,a) \evaleps{2}_{\LL} a$,
 $\foo(f(b,a),a) \evaleps{3}_{\LL} b$.
We also have
 $\foo(\foo(a,a),\foo(f(b,a),a)) \evaleps{3}_{\LL} \foo(\foo(a,a),b) \evaleps{2}_{\LL}  \foo(a,b) \evaleps{1}_{\LL} b$.
\end{example}

The semantics preservation guaranteed by $\RrCat$ has several
consequences. On one hand, the corresponding transformations can be
used to check the completeness of an equational definition and on the
other hand, they can be used to transform extended patterns into
equivalent constructor patterns (potentially aliased) and eventually
to compile prioritized equational definitions featuring extended
patterns into classical and order independent definitions. The
former was detailed in Section~\ref{se:elimination}. For the latter, we
proceed in several steps and we define first a transformation
$\TRcomp$ which encodes a list of  {\er}s
$\LL=[\foo(\ovr{p_1})\ra t_1,\ldots,\foo(\ovr{p_n})\ra t_n]$ into a
list of rules using only constructor patterns potentially aliased:

\[
\begin{array}{l@{\hspace{3pt}}l@{\hspace{3pt}}}
\TRcomp(\LL)=\oplus_{i=1}^{n} 
[ & 
\foo(\ovr{q_1^i}) \ra t_i \otrsep \ldots \otrsep \foo(\ovr{q_m^i})\ra t_i \mid \\
&
\ovr{q_1^i}\plus\ldots\plus\ovr{q_m^i}=\ovr{p_i}\downarrow_{\RrCat}\neq\bot,\\
& 
\ovr{q_1^i},\ldots,\ovr{q_m^i} \text{ contain no}\text{ symbol } \plus
]
\end{array}
\]~\\
The order should be preserved in the resulting list of rules,
{\ie} the rules obtained for a given {\er} should be placed in
the resulting list at the corresponding position. The transformation
can be applied in a similar way to sets of {\er}s.

Note that if the left-hand side of a rule reduces to $\bot$ then,
there is no corresponding rule in the result of the
transformation.
According to Proposition~\ref{prop:mvarPreserve},
$\mvar{p_i}\subseteq\mvar{q_j^i}$ for all $j\in[1,\ldots,m]$ and since
$\RrCat$ preserves the linearity then, the rules in
$\TRcomp(\LL)$ are all well-formed.

\begin{example}
\label{ex:TRcomp}
We consider the signature and  the list of {\er}s in Example~\ref{ex:equationLists}.
The pattern $\foo(x, y\at\ap a)$ of the first {\er} reduces {\wrt}
$\RrCat$ to ${f(x,y\at b)}\plus {f(x,y\at f(y_1,y_2))}$ (reductions are
similar to those in Example~\ref{ex:antip} but we also consider
aliasing now).  The pattern $\foo(a \plus b,y)$ reduces immediately to
the pattern $\foo(a,y) \plus \foo(b,y)$.
We have thus
$\TRcomp(\LL)=[
{f(x,y\at b) \ra y}
\otrsep 
{f(x,y\at f(y_1,y_2)) \ra y}
\otrsep $\linebreak
$
{\foo(a,y) \ra y}
\otrsep 
{\foo(b,y) \ra y}
\otrsep{\foo(f(x,y),z) \ra x}
]$.

We can then check that the reductions presented in
Example~\ref{ex:equationLists} are still possible with $\TRcomp(\LL)$:
 $\foo(a,b) \evaleps{1}_{\TRcomp(\LL)} b$,
 $\foo(a,a) \evaleps{3}_{\TRcomp(\LL)} a$,
 $\foo(b,a) \evaleps{4}_{\TRcomp(\LL)} a$ and
$\foo(f(b,a),a) \evaleps{5}_{\TRcomp(\LL)} b$.
As before, we have also
 $\foo(\foo(a,a),\foo(f(b,a),a)) \evaleps{5}_{\TRcomp(\LL)} \foo(\foo(a,a),b) \evaleps{3}_{\TRcomp(\LL)}  \foo(a,b) \evaleps{1}_{\LL} b$.
\end{example}

The transformation preserves the corresponding relations:
\begin{restatable}[Complement encoding]{prop}{compEnc}
  \label{prop:compEnc}
  Given a list of {\er}s $\LL$ and a term $t\in\TF$, we have $t\eval_{\LL}t'$ iff
  $t\eval_{\TRcomp(\LL)}t'$.
\end{restatable}

The left-hand sides of the rules obtained following the $\TRcomp$
transformation are constructor patterns potentially aliased and to
remove the aliases from these rules and replace accordingly the
concerned variables in the corresponding right-hand sides we use the
following recursive transformation:
\[
\begin{array}{r@{\hspace{3pt}}c@{\hspace{3pt}}l}
\TRat(\foo(p)\ra r) & = & \foo(p)\ra r\\
   && \text{if }\forall\omega\in\PPos(p),p(\omega)\not=\at \\
   \TRat(\foo(\rmp{p}{\omega}{x\at q})\ra r) & = & \TRat(\foo(\rmp{p}{\omega}{q})\ra\{x\mapsto q\}{r}) \\
   && \text{if }\forall \omega'\in\PPos(q),q(\omega')\not=\at
\end{array}
\]~\\
which extends to lists ans sets of {\er}s: 
\[
\TRat([e_1,\ldots,e_n])=[\TRat(e_1),\ldots,\TRat(e_n)]
\]
\[
\TRat(\{e_1,\ldots,e_n\})=\{\TRat(e_1),\ldots,\TRat(e_n)\}
\]

Note that since at each intermediate transformation step the
considered aliased pattern $q$ contains no aliases itself then the
right-hand sides $\{x\mapsto q\}{r}$ of the obtained rules contain
no aliases at the positions concerned by the replacement and thus,
become eventually terms in $\TFX$.

\begin{example}
\label{ex:TRcompAT}
We consider the list of rules 
obtained by applying the transformation $\TRcomp$ in
Example~\ref{ex:TRcomp}: $\TRcomp(\LL)=[
{f(x,y\at b) \ra y}
\otrsep
{f(x,y\at f(y_1,y_2)) \ra y}
\otrsep$
${\foo(a,y) \ra y}
\otrsep 
{\foo(b,y) \ra y}
\otrsep 
\foo(f(x,y),z) \ra x
]$.
We have then $\TRat(\TRcomp(\LL))=[
f(x,b) \ra b
\otrsep 
f(x,f(y_1,y_2)) \ra f(y_1,y_2)
\otrsep 
\foo(a,y) \ra y
\otrsep 
\foo(b,y) \ra y
\otrsep 
\foo(f(x,y),z) \ra x
]$.
\end{example}

The transformation $\TRat$ obviously preserves the ground semantics of
the left-hand side of the transformed rules and terminates since
it decreases the number of aliases in the left-hand side.  The result
of the transformation preserves the one-step semantics of the original
rules:
%
\begin{restatable}[Alias encoding]{prop}{aliasEncoding}
\label{prop:aliasEncoding}
  Given a list of rules $\LL$ using only constructor as-patterns
  and a term $t\in\TF$, we have $t\eval_{\LL}t'$ iff
  $t\eval_{\TRat(\LL)}t'$.
\end{restatable}

We have thus a method allowing the transformation of a list of
rules using extended patterns, {\ie} complement patterns (and thus
anti-patterns), sum patterns and as-patterns, into an equivalent list
using only constructor patterns.  This is particularly useful when we
have to compile towards languages which can handle only this kind of
patterns. We also have the ingredients to encode it into a set of
rules; this is interesting when we want to use (reasoning) tools
for which is more convenient to handle order independent rules.
For this, we replace the pattern of each rule in a list by
complementing it {\wrt} the  previous rules in the list.
More precisely, given a list of rules $\LL=[{\foo(\ovr{p_1})\ra
t_1},\ldots,\foo(\ovr{p_n})\ra t_n]$ we consider the transformation
\[
\begin{array}{l@{\hspace{3pt}}l@{\hspace{3pt}}l}
\TRord(\LL)=\cup_{i=1}^{n} 
\{ &
\foo(\ovr{q_1^i}) \ra t_i \otrsep \ldots \otrsep \foo(\ovr{q_m^i})\ra t_i \mid\\
&
\ovr{q_1^i}\plus\ldots\plus\ovr{q_m^i}=\ovr{p_i}\minus(\ovr{p_1}\plus\ldots\plus \ovr{p_{i-1}})\downarrow_{\RrCat}\neq\bot,\\
&
\ovr{q_1^i},\ldots,\ovr{q_m^i} \text{ contain no}\text{ symbol } \plus \}
\end{array}
\]
which preserves the initial relation:
\begin{restatable}[Order encoding]{prop}{orderEnc}
\label{prop:orderEnc}
  Given a list of {\er} $\LL$ and a term $t\in\TF$, we have
  $t\eval_{\LL}t'$ iff $t\eval_{\TRord(\LL)}t'$.
\end{restatable}

\begin{example}
\label{ex:appPlus}
We consider the signature in Example~\ref{ex:TRcomp} and the list of
rules $\LL=[{\foo(z,a) \ra z} \otrsep$ ${\foo(x,y) \ra y}]$.  The
pattern $\tup{x,y}\minus \tup{z,a}$ reduces {\wrt} $\RrCat$ to
$f(x,y\at b)\plus f(x,y\at f(y_1,y_2))$ (reductions are similar to
those in Example~\ref{ex:appComp} and Example~\ref{ex:simple} but we
also consider aliasing now).  We have thus $\TRord(\LL)=\{\foo(z,a)
\ra z \otrsep {\foo(x,y\at{b})} \ra y \otrsep \foo(x,y\at f(y_1,y_2)) \ra
y\}$ and $\TRat(\TRord(\LL))=\{{\foo(z,a) \ra z} \otrsep$ ${\foo(x,b) \ra b}
\otrsep$ $\foo(x,f(y_1,y_2)) \ra f(y_1,y_2)\}$.
\end{example}

The transformations presented in this section preserve the one step
semantics of the transformed lists (or sets) of rules.  Moreover,
the transformations are well-defined, {\ie} produce lists or sets of
well-formed rules. The produced rules have a specific shape
suitable for subsequent transformations. We can thus combine them to
define a transformation $\TRall=\TRat\circ\TRcomp$ which transforms a
list of rules using extended patterns into one using only
constructor patterns, or a transformation $\TRorder=\TRat\circ\TRord$
which transforms a list of rules (using extended patterns) into
a set of rules using only plain constructor patterns, {\ie} a CTRS.

\begin{corollary}[Simulation]
  \label{th:simulation}
  Given a list of {\er}s $\LL$ and a term $t\in\TF$, 
  \begin{enumerate}
  \item $t\eval_{\LL}t'$ iff
    $t\eval_{\TRall(\LL)}t'$;
  \item $t\eval_{\LL}t'$ iff
    $t\eval_{\TRorder(\LL)}t'$.
\end{enumerate}
\end{corollary}

\newcommand{\sat}[1]{\overline{#1}}

\section{Optimizations and experimental results}
\label{se:experiments}
%
We know that our algorithms can exhibit exponential time behavior
since the useful clause problem is NP-complete~\cite{Sekar-1992}.
Since we target practical implementations with reasonable running
times, we identify the critical aspects leading to this exponential
behavior and try to limit as much as possible this explosion.
The exponential behavior originates in our case from the rules
$\rmMinusVM$ and $\rmMinusFF$ (Figure~\ref{fig:elimComplement}) for
which the number of elements in the generated sums determines the
effective branching.
We propose two optimizations that turn out to limit significantly this
number for concrete examples.

\smallbreak\noindent\textbf{Cut useless choices.}  
For a symbol~$f$ of arity~$n>0$, the rule~{\rmMinusFF} transforms the
term $f(\meta{v_1},\ldots,\meta{v_n}) \minus
f(\meta{t_1},\ldots,\meta{t_n})$ into a sum
$\sum_{i=1}^n{f(\meta{v_1},\ldots,\meta{v_i}\minus
  \meta{t_i},\ldots,\meta{v_n})}$ of~$n$ new terms to reduce.  We can
remark that if there exists a~$k$ such that $\meta{v_k}\minus
\meta{t_k}=\meta{v_k}$ then the $k$-th term of the sum is the term
$f(\meta{v_1},\ldots,\meta{v_n})$ which subsumes all the other terms
in the sum and whose semantics is thus the same as that of the sum.
Therefore, as soon as such a term is exhibited the sum can be
immediately reduced to $f(\meta{v_1},\ldots,\meta{v_n})$ avoiding thus
further unnecessary reductions.
For example, the term $f(x,a)\minus f(a,b)$ normally reduces to
$f(x\minus a,a)\plus f(x,a\minus b)$ which is eventually reduced to $f(x,a)$,
while using the optimization we get directly 
$f(x,a)$.

\smallbreak\noindent\textbf{Sorted encoding.}
Given a term of the form $\varVar \minus
g(\meta{t_1},\ldots,\meta{t_n})$, the rule~{\rmMinusVM} produces a sum
$\sum_{f\in\CC} f(\varCt_1,\ldots,\varCt_{m}) \minus
g(\meta{t_1},\ldots,\meta{t_n})$ containing an element for each
constructor~$f$ of the signature.  In practice, algebraic signatures
are often many-sorted and in this case,
since we can always identify the sort of the variable $\varVar$ in a
given context then, the sum $\sum_{f\in\FF}
f(\varCt_1,\ldots,\varCt_{m})$ can be restricted to all the
constructors
of this sort.

\begin{example}
  Let us consider the many-sorted signature $E=a\mid b\mid c$ and
  $L=cons(E,L)\mid nil$. The application of rule~{\rmMinusVM} to the
  term $cons(z\minus a,nil)$ produces $cons(z\at((a\plus b\plus c\plus
  nil\plus cons(z_1,z_2))\minus a),nil)$ which is reduced, by
  propagation of $\minus$, to $cons(\varCt\at(a\minus a \plus b\minus
  a\plus c\minus a\plus nil\minus a\plus cons(\varCt_1,\varCt_2)\minus
  a),nil)$. This term contains ill-typed terms like $nil\minus a$ or
  $cons(z_1,z_2)\minus a$ that are eventually reduced to $\bot$ and
  eliminated.  With the optimization, we infer the type~$E$ for $z$
  and generate directly the correctly typed term $cons(z\at((a\plus
  b\plus c)\minus a),nil)$.
\end{example}

\smallbreak\noindent\textbf{Improved minimization.}
%
The minimization algorithm has also an exponential complexity on the number
$n$ of input rules.
A first optimization follows the observation that if a pattern $l_i$
subsumes a pattern $l_j$ then, $l_j$ cannot be in the minimal set of
patterns.  We can thus safely eliminate all patterns directly subsumed
by another one right from the beginning and decrease the number of
recursive calls accordingly.

A second optimization consists in initializing the kernel with all the
patterns~$l_i$ which are not subsumed by
$l_1,\ldots,l_{i-1},l_{i+1},\ldots,l_n$.  Indeed, for such~$l_i$ the
test ``$l_i$ is subsumed by $S\cup kernel$'' 
is always false and thus, $l_i$ is added systematically to the kernel
during the computation. Initializing the kernel with these $l_i$
reduces the complexity from $O(2^n)$ to $O(2^{n-k})$, where~$k$ is the
number of such patterns.

\smallbreak\noindent\textbf{Local minimum  vs.  global minimum.}
Given a set of {\plaint}s~$S$, the algorithm given in
Figure~\ref{fig:minimize} computes the smallest subset $S'\subseteq S$
such that $\validsubset{S'}{S}$.  In the general case $S'$ may not be
the smallest set such that $\semgset{S'}=\semgset{S}$ and we show how
our algorithm can be used to find this smallest set.

We consider the \emph{saturation} $\sat{S}$ of a set of {\plaint}s $S$
{\wrt} to all {\plaint}s subsuming terms in $S$ and preserving the
semantics:
$\sat{S}=\cup_{q\in S}\{p \mid\semg{q}\subseteq\semg{p}\subseteq\semg{S} \}$.
We can show that the minimum of the saturated set is smaller than the
minimum of the original one, $\card{\minimum(\sat{S})}\leq\card{\minimum(S)}$, and that the minimum of the saturations of two
sets of patterns with the same semantics $\semg{S}=\semg{S'}$ is 
the same, $\minimum(\sat{S})=\minimum(\sat{S'})$.
We can thus, take any of them and compute its saturation and
the corresponding (local) minimal set of patterns. Since the minimum
of a saturated set is smaller than any of its subsets then, the
obtained minimum is global:
\begin{proposition}
  $\minimum(\sat{S})$ is the (global) minimum valid subset of $S$.
\end{proposition}

This is not an actual optimization but just an extension which
guarantees the global minimality. For most of the examples we have
experienced with, this global minimization technique had no impact,
the minimum set of patterns being obtained directly by our rule
elimination.

\smallbreak\noindent\textbf{Implementation.}  
%
All the transformations and optimizations presented in the paper have
been implemented in~{\tom}, an extension of {\java} allowing the use
of rules and strategies. This
implementation\footnote{\url{https://github.com/rewriting/tom},
  \texttt{[scm]/applications/strategyAnalyzer/}}
can generate TRSs expressed in several syntaxes like, for example,
{\aprove}\cite{terminationDP-aprove2011}/{\TTT}\cite{ttt-2009}  syntax which can be used to check termination, and
{\tom} syntax which can be used to execute the resulting TRS.

An alternative Haskell implementation\footnote{\url{http://github.com/polux/subsume}} has allowed 
us to generate via GHCJS a javascipt version of the algorithm that can
be experimented in a
browser
\footnote{\url{http://htmlpreview.github.io/?https://github.com/polux/subsume/blob/web/out/index.html}}.

For simplicity, we presented the formalism in a mono-sorted context
but both implementations handle many-sorted signatures and implement
the corresponding optimization explained above.

\section{Related works}
\label{se:related}
%
A.~Krauss studied the problem of transforming function definitions
with pattern matching into minimal sets of independent
equations~\cite{ICFP-2008-Krauss}. Our approach can be seen as a new
method for solving the same problem but arguably easier to implement because of
the clear and concise rule based specification. 
We handle here the right-hand sides of the rewrite rules and allow
anti-patterns in the left-hand sides of rules and although this seems
also feasible with the method in~\cite{ICFP-2008-Krauss} the way it
can be done is not made explicit in the paper.
We couldn't obtain the prototype mentioned in the paper but when
experimenting with the proposed 
examples
 we obtained execution times which
indicate comparable performances to those given
in~\cite{ICFP-2008-Krauss} and, more importantly, identical
results:
\begin{itemize}
\item \textsf{Interp} is an interpreter for an expression language
  defined by 7 ordered rules for which a naive disambiguation would
  produce 36 rules~\cite{ICFP-2008-Krauss}.  Our transformation
  without minimization
  produces 31 rules and using 
  the minimization algorithm presented in Section~\ref{se:elimination}
  we obtain, as in~\cite{ICFP-2008-Krauss}, 25 rules.
  Example~\ref{ex:elimPattern} is indeed inspired by one of the rules
  eliminated during the minimization process for this specification.
\item \textsf{Balance} is a balancing function for red-black-trees.
  A.~Krauss reported that for this list
  of  5 ordered rules a naive approach would produce 91 rules.  In our
  case, 
  the transformation directly produces  the minimal set composed of 59
  rules reported in~\cite{ICFP-2008-Krauss}.  
\item \textsf{Numadd} is a function that operates on arithmetic
  expressions. It is composed of 5 rules.  
  Our transformation
  directly produces the minimal set composed of 256 rules~\cite{ICFP-2008-Krauss}.
\end{itemize}

L.~Maranget has proposed an algorithm for detecting useless
patterns~\cite{JFP:977988} for OCaml and Haskell. 
As mentioned previously, the algorithm in
Figure~\ref{fig:elimComplement} can be also used to check whether a
pattern is useless {\wrt} a set of patterns
(Proposition~\ref{th:subsumption}) but it computes the difference
between the pattern and the set in order to make the decision.
The minimization algorithm in Figure~\ref{fig:minimize} can thus use
the two algorithms interchangeably.
Both algorithms have been implemented and we measured the execution
time for the minimization function
on various examples.  In
average, L.~Maranget's approach is $30\%$ more efficient than ours
(40ms vs 55ms for the \textsf{Interp} example for instance) and can
thus be used in the minimization algorithm if one wants to gain some
efficiency with the price of adding an auxiliary algorithm.

This work has been initially motivated by our encoding of TRSs guided
by rewriting strategies into plain
TRSs~\cite{DBLP:conf/rta/CirsteaLM15}. This encoding produces
intermediate systems of the form $\{ \varphi(l)\!\ra\!r,
\varphi(x\at\ap l)\!\ra\!r',$ $\ldots \}$ which are eventually reduced
by expanding the anti-patterns into plain TRSs.
Alternatively, we could use a simpler compilation schema based on ordered CTRSs in
which case the intermediate system would have the form $[
\varphi(l)\ra r \otrsep \varphi(x)\ra r' \otrsep \ldots ]$ and then,
apply the approach presented in this paper to transform the
resulting ordered CTRS into an order independent CTRS.
We experimented with this new approach and for all the examples
in~\cite{DBLP:conf/rta/CirsteaLM15} we obtained between $20\%$ and
$25\%$ less rules than before.
When using strategies, the order of rule application is expressed with
a left-to-right strategy choice-operator and for such strategies the
gain with our new approach is even more significant than for the
examples in~\cite{DBLP:conf/rta/CirsteaLM15} which involved only
plain, order independent, TRSs.

There are a lot of
works~\cite{toplas-GieslRSST11,KraussSTFG11,AvanziniLM15,GieslBEFFOPSSST14}
targeting the analysis of functional languages essentially in terms of
termination and complexity, and usually they involve some encoding of
the match construction. These encodings are generally deep and take
into account the evaluation strategy of the targeted language leading
to powerful analyzing tools. Our encodings are shallow and independent
of the reduction strategy.
Even if it turned out to be very practical for encoding ordered CTRSs involving
anti-patterns and prove the (innermost) termination of the
corresponding CTRSs with {\aprove/\TTT}, in the context of functional
program analysis we see our approach more like a helper that will be
hopefully used as an add-on by the existing analyzing tools.
%

\section{Conclusion}
\label{se:conclusion}

We have proposed a concise and clear algorithm for computing the
complement of a pattern {\wrt} a set of patterns and we showed how it
can be used to encode an OTRS potentially containing anti-patterns
into a plain TRS preserving the one-step semantics of
the original system.
The approach can be used as a generic compiler for ordered rewrite
systems involving anti-patterns and, in collaboration with
well-established techniques for TRS, for analyzing properties of such
systems. 
Since the TRSs obtained with our method define exactly the same
relation over terms, the properties of the TRS stand also for the
original OTRS and the counter-examples provided by the analyzing tools
when the property is not valid can be replayed directly for the OTRS.
Moreover, our approach can be 
used as a new method for detecting
useless patterns and for the minimization of sets of patterns.

For all the transformations we have performed the global minimization
technique was superfluous and we conjecture that, because of the shape
of the problems we handle and of the way they are handled, the
transformation using a local minimization
produces directly the smallest TRS for any input OTRS. One of our
objectives is to prove this conjecture.

We consider of course integrating our algorithm into automatic tools
either to disambiguate equational specifications or to (dis)prove
properties of such specifications. 
The two available implementations let us think that such an
integration can be done smoothly for any tool relying on a declarative
programming language.






\smallbreak\noindent\textbf{Acknowledgments.}
We want to thank Paul Brauner who implemented the Haskell version of
the algorithm and who gave us very helpful feedback on previous drafts
of this paper.

\bibliographystyle{abbrv}
\bibliography{paper}

\appendix
\section{Proofs}
\label{ap:proofs}

\subsumptionVSsemantics*
\begin{proof}
  By induction on the structure of $p$. 
  If $p$ is a variable then, for
  any value $v$ we have ${v}\in\semg{{p}}$ and ${p} \matcha {v}$.
  Consider now that $p=c(p_1,\ldots,p_n)$ with $c\in\CC\cup\CCt$ and
  $p_1,\ldots,p_n$ patterns. 
  If ${v}\in\semg{p}$ then $v=c(v_1,\ldots,v_n)$ with
  $(v_1,\ldots,v_n)\in\semg{p_1}\times\ldots\times\semg{p_n}$ and by
  induction ${p_i} \matcha {v_i}$ for all $i\in[1,\ldots,n]$. Thus, by
  definition of the instance relation ${p} \matcha {v}$.
  If ${p} \matcha {v}$ then $v=c(v_1,\ldots,v_n)$ with $p_i \matcha
  v_i$ for all $i\in[1,\ldots,n]$ and by induction $v_i\in\semg{p_i}$
  for all $i\in[1,\ldots,n]$. Thus, by definition of ground semantics
  ${v}\in\semg{{p}}$.
\end{proof}

\subsumptionVSsemanticsEXT*
\begin{proof}
  By induction on the structure of $p$. 
  If $p$ is a variable then, for
  any value $v$ we have ${v}\in\semg{{p}}$ and ${p} \matcha {v}$.
  If $p$ is a plain pattern we proceed as for
  Proposition~\ref{th:subsumptionVSsemantics}.

  If $p={p_1}\plus{p_2}$ then $v\in\semg{p_1}\cup\semg{p_2}$ and thus
  $v\in\semg{p_1}$ or $v\in\semg{p_2}$. By induction ${p_1} \matcha
  {v}$ or ${p_2} \matcha {v}$ and thus, ${{p_1}\plus{p_2}} \matcha
  {v}$. 
  For the other direction we consider that ${{p_1}\plus{p_2}} \matcha
  {v}$ and then $v\in\semg{p_1} \vee v\in\semg{p_2}$. By induction,
  $v\in\semg{p_1}$ or $v\in\semg{p_2}$ and thus,
  $v\in\semg{p_1}\cup\semg{p_2}=\semg{p_1\plus p_2}$.

  For the case $p={p_1}\minus{p_2}$ we proceed similarly but use the
  corresponding set properties. 
  If $p={p_1}\minus{p_2}$ then $v\in\semg{p_1}\setminus\semg{p_2}$ and thus
  $v\in\semg{p_1}$ and $v\not\in\semg{p_2}$. By induction ${p_1} \matcha
  {v}$ and ${p_2} \nmatcha {v}$ and thus, ${{p_1}\minus{p_2}} \matcha
  {v}$. 
  For the other direction we consider that ${{p_1}\minus{p_2}} \matcha
  {v}$ and then $v\in\semg{p_1} \vee v\not\in\semg{p_2}$. By induction,
  $v\in\semg{p_1}$ and $v\not\in\semg{p_2}$ and thus,
  $v\in\semg{p_1}\setminus\semg{p_2}=\semg{p_1\minus p_2}$.

  The case where $p=\bot$  is obvious. 
\end{proof}

For the termination proof of the rewriting systems $\RrC$ and $\RrCat$
we have used two approaches. The first one consists in providing a
meta-encoding (the complete one for $\RrCat$ is provided in
Appendix~\ref{ap:metaEncoding}) of a complete approximation of the
rule schemas of $\RrC$ and $\RrCat$: to any step in the reduction
relation induced by $\RrCat$, and thus $\RrC$, correspond one or
several reduction steps in the approximation (which can also contain
other reduction paths) and consequently, the termination of the
over-approximation implies the termination of $\RrC$ and
$\RrCat$. Both {\TTT} and {\aprove} managed to find a proof for the
system implementing the meta-encoding. This method has the advantage
of providing an automatic proof.
We also provide below a hand proof of the property.

\convergence*
\begin{proof}

  For termination, we use the lexicographic order over a reduction
  ordering and the recursive path order (RPO) with
  status~\cite{Schneider-KampTACG07} defined below. Intuitively, these
  two orders correspond respectively to the two characteristics decreasing during the
  reduction with the rewriting system: the number of complements
  potentially concerning a variable and the height of the complements
  in the term. 

  The reduction ordering $\sizeorder$ is decreasing on all rules and
  strictly decreasing on $\rmMinusVM$. The RPO $\rpo$ is strictly
  decreasing on all rules but $\rmMinusVM$.

  For the ordering $\sizeorder$ we use a mapping $\sizep$ which
  intuitively maps terms to multisets of positions potentially
  concerned by a complement {\wrt} a variable ({\ie} $X\minus u$).
  We use the strict inclusion $\sizeorder$ as an order over such sets.\\[+3pt]
  $
  \begin{array}{l@{\hspace{10pt}}l@{\hspace{10pt}}l@{\hspace{10pt}}l}
    \sizep(X) = \emptyset,  \hfill  \forall X\in\XX\\
    \sizep(f(u_1,\ldots,u_n)) = \sizep({u_1}) \cup \ldots \cup \sizep({u_n}),  \hfill     \forall f\in\FF^{n} \\
    \sizep({u_1}\plus{u_2}) = \sizep({u_1}) \cup \sizep({u_2})  \\
    \sizep({u_1}\minus{u_2}) = \sizep({u_1}) \cup \sizep({u_2}) \cup (\varp(u_1) \cap \symp(u_2)) \\
    \sizep(\bot) = \emptyset \\
  \end{array}
  $\\[+3pt]
  with\\[+3pt]
  $
  \begin{array}{l@{\hspace{10pt}}l@{\hspace{10pt}}l@{\hspace{10pt}}l}
    \varp(X) = \allPos & 
    \forall X\in\XX
    \\
    \varp(f(v_1,\ldots,v_n)) = \cup_{i=1}^{n}\{i.p \mid p\in\varp(v_i) \} &
    \forall f\in\FF^{n}
    \\
    \varp({v_1}\plus{v_2}) = \varp({v_1}) \cup \varp({v_2}) & 
    \\
    \varp({v_1}\minus{v_2}) = \allPos & 
    \\
    \varp(\bot) = \emptyset & 
    \\
    \symp(X) = \emptyset &
    \forall X\in\XX
    \\
    \symp(f(v_1,\ldots,v_n)) = \cup_{i=1}^{n}\{i.p \mid p\in\symp(v_i) \} \cup \varepsilon &
    \forall f\in\FF^{n}
    \\
    \symp({v_1}\plus{v_2}) = \symp({v_1}) \cup \symp({v_2}) &
    \\
    \symp({v_1}\minus{v_2}) = \allPos &
    \\
    \symp(\bot) = \emptyset &
  \end{array}
  $\\[+3pt]
  with $\allPos$ the set of all positions.\\
  Note that $\sizep(v)=\emptyset$ for any {\addt} term $v$ ({\ie}
  containing no $\minus$).  

  To define the RPO order $\rpo$ we explicit the use of pattern
  variables of the form $\varVar$ using a special symbol $\patvar$ and
  consequently these variables are now of the form $\patvar(id)$ with
  $id$ the name of the variable.  The quasi-precedence used for this
  RPO order $>$ is then defined as follows:
  $\plus,\bot,f_n,\patvar\rpo\minus$ 
  for all $f_n\in\FF^{n}$ with
  statuses $\minus:lexicographic$ (permutation $[2,1]$),
  $\patvar:lexicographic ([1])$,\linebreak 
  $\plus:lexicographic ([1,2])$, 
  $\bot:lexicographic$,
  $f_n:lexicographic ([1,$ $\ldots,n])$.
  It is easy to check that for each rule ${l}\raM{r}$ in $\RrC$
  we have $l>r$.


For the shape of the normal forms {\wrt} this system we proceed by
contradiction.
We suppose that a term of the form $u=u'\minus u''$ with $u',u''$
additive terms is a normal form.
$u''$ 
cannot be a variable since $u$ would be a redex for $\rmMinusMV$, 
cannot be a $\bot$ since $u$ would be a redex for $\rmMinusMB$, 
cannot be a $\plus$ since $u$ would be a redex for $\rmMinusMP$.
If $u''$ is a term of the form $f(u_1,\ldots,u_n)$ then $u'$
cannot be a variable since $u$ would be a redex for $\rmMinusVM$, 
cannot be a $\bot$ since $u$ would be a redex for $\rmMinusBM$, 
cannot be a term whose head symbol is a $\plus$ since $u$ would be a redex for $\rmMinusPM$, 
cannot be of the form $h(u'_1,\ldots,u'_m)$ since $u$ would be a redex for
$\rmMinusFG$ or $\rmMinusFF$. Thus the normal forms contain no
complements.

Suppose now the normal form contains a sum just below a constructor,
{\ie} the term is of the form $f(u_1,\ldots,u'\plus u'',\ldots,u_n)$;
this is not possible since the term would be a redex for
$\rmDistAdd$. Thus the normal forms contains no sum below a
constructor.

We consider now a normal form $u$ containing a $\bot$ at a position other than
the root one. $u$ 
cannot be of the form $\bot\plus u'$ since it would be a redex for $\rmAddL$, 
cannot be of the form $u'\plus\bot$ since it would be a redex for $\rmAddR$, 
cannot be of the form $f(u_1,\ldots,\bot{},\ldots,u_n)$ since it would be a redex for $\rmEmptyF$, 
cannot be of the form $\bot\minus u'$ or $u'\minus\bot$ as seen
above. Thus, normal forms can contain $\bot$ only at the root position.


We show the local confluence of the system by proving that all
critical pairs induced by the rewrite rules of the system converge.
  We have the following critical pairs which all converge:\\
  $(\rmAddL)-(\rmAddR)$ (converge directly),\\
  $(\rmAddL)-(\rmDistAdd)$ and $(\rmAddR)-(\rmDistAdd)$ (converge with  $\rmEmptyF$ and $(\rmAddL)$/$(\rmAddR)$),\\
  $(\rmAddL)-(\rmMinusMP)$ and $(\rmAddR)-(\rmMinusMP)$ (converge with $\rmMinusMB$),\\
  $(\rmAddL)-(\rmMinusPM)$ (converge with $\rmMinusBM$ and $\rmAddL$),\\
  $(\rmAddR)-(\rmMinusPM)$ (converge with $\rmMinusBM$ and $\rmAddR$),\\
  $(\rmEmptyF)-(\rmDistAdd)$ (converges with  twice $\rmEmptyF$, $\rmAddL$),\\
  $(\rmEmptyF)-(\rmMinusBM)$ (converges with $\rmMinusMB$),\\
  $(\rmEmptyF)-(\rmMinusPM)$ (converges with $\rmMinusMB$ and twice  $\rmEmptyF$, $\rmMinusMB$),\\
  $(\rmEmptyF)-(\rmMinusFF)$ only left possible (converges with $\rmMinusMB$ and $\rmMinusMB$, $n$ times $(\rmEmptyF)$, $n$ times $(\rmAddL)$/$(\rmAddR)$)\\
  $(\rmEmptyF)-(\rmMinusFG)$ left (converges with $\rmMinusBM$ and $\rmEmptyF$),\\
  $(\rmEmptyF)-(\rmMinusFG)$ right (converges with $\rmMinusMB$),\\
  $(\rmDistAdd)-(\rmMinusVM)$ (converges with $\rmMinusMP$, $\rmMinusVM$ and $\rmDistAdd$, $\rmMinusMP$),\\
  $(\rmDistAdd)-(\rmMinusBM)$ (converges with $\rmMinusMP$, twice $\rmMinusBM$),\\
  $(\rmDistAdd)-(\rmMinusPM)$ (converges with $\rmMinusPM$ and twice  $\rmDistAdd$),\\
  $(\rmDistAdd)-(\rmMinusFF)$ left (converges with $\rmMinusPM$, $\rmMinusFF$ and $\rmMinusPM$, $\rmDistAdd$)\\
  $(\rmDistAdd)-(\rmMinusFF)$ right (converges with $\rmMinusMP$, twice $\rmMinusFF$ and $\rmMinusMP$)\\
  $(\rmDistAdd)-(\rmMinusFG)$ left (converges with $\rmMinusPM$ and $\rmDistAdd$),\\
  $(\rmDistAdd)-(\rmMinusFG)$ right (converges with $\rmMinusMP$, twice $\rmMinusPM$)\\
\end{proof}

\comEq*
\begin{proof}
  We prove that the ground semantics of the left-hand side and right-hand
  side of the rewrite rules of $\RrC$ are the same. 

  For rules $(\rmEmptyF)$ and $(\rmDistAdd)$ we use the fact that
  patterns are linear and apply the recursive definition for
  constructors: 
  $\semg{f(v_1,\ldots,v_n)}=\{c(u_1,\ldots,u_n)\mid
  (u_1,\ldots,u_n)\in\semg{v_1}\times\ldots\times\semg{v_n}\}$.
  Since the semantics of $\bot{}$ is an empty set then so is the
  semantics of $h(\meta{v_1},\ldots,\bot{},\ldots,$ $\meta{v_n})$ for any
  constructor $h$ and thus the property is verified for rule
  $(\rmEmptyF)$.
  For rule $(\rmDistAdd)$ we have
  $\semg{h(v_1,\ldots,\meta{v_i'}\plus\meta{v_i''},\ldots,v_n)}=
  \{c(u_1,\ldots,u_n)\mid (u_1,\ldots,u_n)\in\semg{v_1}\times\ldots\times\semg{\meta{v_i'}\plus\meta{v_i''}}\ldots\times\semg{v_n}\}=
  \{c(u_1,\ldots,u_n)\mid (u_1,\ldots,u_n)\in\semg{v_1}\times\ldots\times\semg{\meta{v_i'}}\cup\semg{\meta{v_i''}}\ldots\times\semg{v_n}\}=
  \{c(u_1,\ldots,u_n)\mid (u_1,\ldots,u_n)\in\semg{v_1}\times\ldots\times\semg{\meta{v_i'}}\ldots\times\semg{v_n}\}
  \cup\{c(u_1,\ldots,u_n)\mid (u_1,\ldots,u_n)\in\semg{v_1}\times\ldots\times\semg{\meta{v_i''}}\ldots\times\semg{v_n}\}=
  \semg{h(v_1,\ldots,\meta{v_i'},\ldots,v_n)}\cup\semg{h(v_1,\ldots,\meta{v_i''},\ldots,v_n)}$

  For the other rules we can apply straightforwardly the definition of
  ground semantics of extended terms; the proof is more elaborate only
  for the rules $(\rmMinusMP)$ and $(\rmMinusFF)$.
  For $(\rmMinusVM)$ we can notice that $\semg{\varVar}$ and
  $\semg{c_1(\varCt_1,\ldots,\varCt_{n_1})\plus\ldots\plus
    c_k(\varCt_1,\ldots,\varCt_{n_k})}$, with
  $\{c_1,\ldots,c_k\}=\CC$, are both equal to $\TC$.

  For rule $(\rmMinusMP)$, we have
  $\semg{\wm \minus (\vm_1 \plus \vm_2)}$=
  $\semg{\wm} \setminus \semg{(\vm_1 \plus \vm_2)}$=
  $\semg{\wm} \setminus (\semg{\vm_1} \cup \semg{\vm_2})$
  and
  $\semg{(\wm \minus \vm_1) \minus \vm_2}$=
  $\semg{(\wm \minus \vm_1)} \setminus \semg{\vm_2}$=
  $(\semg{\wm} \setminus \semg{\vm_1}) \setminus \semg{\vm_2}$=
  $\semg{\wm} \setminus (\semg{\vm_1} \cup \semg{\vm_2})$.

  For  $(\rmMinusFF)$, we have
  $V$=
  $\semg{f(\vm_1,\ldots,\vm_n) \minus f(\tm_1,\ldots,\tm_n)}$= 
  $\semg{f(\vm_1,\ldots,$ $\vm_n)} \setminus \semg{f(\tm_1,\ldots,\tm_n)}$
  and since patterns are linear we can use the recursive definition and we get 
  $U$=$\{f(\vm'_1,\ldots,\vm'_n)\mid(\vm'_1,\ldots,\vm'_n)\in\semg{\vm_1}\times\ldots\times\semg{\vm_n}\}$
  $\setminus$
  $\{f(\tm'_1,\ldots,\tm'_n)\mid(\tm'_1,\ldots,\tm'_n)\in\semg{\tm_1}\times\ldots\times\semg{\tm_n}\}$
  and similarly
  $W$=
  $\semg{f(\vm_1\minus \tm_1,\ldots,\vm_n) \plus \ldots \plus f(\vm_1,\ldots,\vm_n\minus \tm_n)}$=
  $\semg{f(\vm_1\minus \tm_1,\ldots,\vm_n)} \cup \ldots \cup \semg{f(\vm_1,\ldots,\vm_n\minus \tm_n)}$=
  $\{f(\vm'_1,\ldots,\vm'_n)\mid(\vm'_1,\ldots,\vm'_n)\in\semg{\vm_1\minus \tm_1}\times\ldots\times\semg{\vm_n}\}$
  $\cup \ldots \cup$
  $\{f(\vm'_1,\ldots,\vm'_n)\mid(\vm'_1,\ldots,\vm'_n)\in\semg{\vm_1}\times\ldots\times\semg{\vm_n\minus \tm_n}\}$.
  We have to show that $U=W$ and
  thus that $\forall s, s\in U \text{ iff } s\in W$.
  Take $s\in U$, then 
  $s\in\{f(\vm'_1,\ldots,\vm'_n)\mid(\vm'_1,\ldots,\vm'_n)\in\semg{\vm_1}\times\ldots\times\semg{\vm_n}\}$ 
  and 
  $s\not\in\{f(\tm'_1,\ldots,\tm'_n)\mid(\tm'_1,\ldots,\tm'_n)\in\semg{\tm_1}\times\ldots\times\semg{\tm_n}\}$  
  and thus, 
  $s=f(\wm_1,\ldots,\wm_n)$ s.t. $\forall i,
  \wm_i\in\semg{\vm_i}$ and $\exists j,
  \wm_j\not\in\semg{\tm_j}$. Consequently, we have
  $\wm_j\in\semg{\vm_j}\setminus\semg{\tm_j}$ and thus,
  $f(\wm_1,\ldots,\wm_n)\in\{f(\vm'_1,\ldots,\vm'_n)\mid(\vm'_1,\ldots,\vm'_n)\in\semg{\vm_1}\times\ldots\times\semg{\vm_j\minus \tm_j}\times\ldots\times\semg{\vm_n}\}$
  $\subseteq W$. 
  We can show similarly that if $s\in W$ then $s\in U$.


\end{proof}

\subsumption*
\begin{proof}
  If $p\minus (p_1\plus\cdots\plus p_n) \multievalM{\RrC} \bot$ then,
  according to Proposition~\ref{prop:comEq}, $\semg{p\minus
    (p_1\plus\cdots\plus p_n)}=\semg{\bot}=\emptyset$. Consequently, we have\linebreak
  $\semg{p}\setminus\bigcup_{i=1}^{n}\semg{p_i}=\emptyset$ and thus
  $\semg{p}\subseteq\bigcup_{i=1}^{n}\semg{p_i}$ so $p$ is
    subsumed by $\{ p_1,\ldots,p_n \}$.

  If $p \text{ is subsumed by } \{ p_1,\ldots,p_n \}$ then
  $\semg{p}\subseteq\bigcup_{i=1}^{n}\semg{p_i}$ and thus
  $\semg{p}\setminus\bigcup_{i=1}^{n}\semg{p_i}=\emptyset$ and
  consequently $\semg{p\minus (p_1\plus\cdots\plus p_n)}=\emptyset$.
  According to Lemma~\ref{lemma:convergence} the normal form of
  $p\minus (p_1\plus\cdots\plus p_n)$ {\wrt} to $\RrC$ is an {\addt}
  term with $\bot{}$ potentially present only at the root position.
  Since the semantics of an {\addt} term with $\bot{}$ potentially
  present only at the root position cannot be empty and since,
  according to to Proposition~\ref{prop:comEq}, $\RrC$ preserves the
  semantics then, the normal form of $p\minus (p_1\plus\cdots\plus
  p_n)$ is necessarily $\bot$ (whose semantics is empty).
\end{proof}

\apEq*
\begin{proof}
  We consider a rewriting {\wrt} $\RrAP$ which occurs at the root position and in this case if
  $p\evalM_{\RrAP}p'$ then there exists a substitution $\tau$
  such that $\tau(\ap\tm)=p$ and since at the meta-level we perform plain
  matching and substitution application this comes to
  $p=\ap\tau(\tm)$. We have thus,
  $p'=\tau(\varCt\minus\tm)=\varCt\minus\tau(\tm)$ and consequently,
  we have to prove that
  $\semg{\ap\tau(\tm)}=\semg{\varCt\minus\tau(\tm)}$.
  For this, it is enough to take $\omega=\varepsilon$ in
  $\semg{\rmp{t}{\omega}{\ap{t'}}}=\semg{\rmp{t}{\omega}{z}}\setminus\semg{\rmp{t}{\omega}{t'}}$
  (since there is no $\omega' < \omega$ s.t.  $t(\omega')=\ap$), and
  we obtain $\semg{\ap{t}}=\semg{z}\setminus\semg{t}=\semg{z\minus
    t}=\semg{\varCt\minus t}$.

  We have, by induction on the structure of terms, that $\forall s,s',
  \text{ if } \semg{s}=\semg{s'}$, then for any linear terms
  $\rmp{t}{\omega}{s},\rmp{t}{\omega}{s'}$ we have
  $\semg{\rmp{t}{\omega}{s}}=\semg{\rmp{t}{\omega}{s'}}$.
  We can thus conclude by definition of one step rewriting and using
  the equivalence for the rewriting at the root position.
\end{proof}


To establish the minimality result for the algorithm in
Figure~\ref{fig:minimize} we first state two properties of the
predicate \emph{``is subsumed by''} which are needed to ensure the
correctness of the algorithm.
%
\begin{lemma}[Valid subset]
  \label{lemma:validsubset}
  Given two sets of
  constructor patterns $S,S'$ such that $S'\subseteq S$, we have:
  $\validsubset{S'}{S}$ iff $\forall q\in S\setminus S'$,
  $q$ is subsumed by~$S'$.
\end{lemma}
\begin{proof}
  $(\Rightarrow)$ When $\validsubset{S'}{S}$ we have $\bigcup_{p\in S'}\semg{p} = \bigcup_{p\in S}\semg{p}$.
  Let $q\in S\setminus S'$, we have $\semg{q}\subseteq\bigcup_{p\in S}\semg{p}$, and thus
  $\semg{q}\subseteq\bigcup_{p\in S'}\semg{p}$. So $q$ is subsumed by~$S'$.

  \noindent
  $(\Leftarrow)$
%
  We proceed by contradiction and we suppose that $\bigcup_{p\in
    S'}\semg{p} \neq \bigcup_{p\in S}\semg{p}$. In this case there
  exists $t\in\bigcup_{p\in S}\semg{p}$ such that
  $t\not\in\bigcup_{p\in S'}\semg{p}$ and thus, we have
  $t\in\bigcup_{p\in S\setminus S'}\semg{p}$. Consequently,
  $\exists q\in S\setminus S'$ such that $t\in\semg{q}$.  By
  hypothesis $q$ is subsumed by~$S'$: we have $\semg{q}\subseteq\bigcup_{p\in S'}\semg{p}$
  and thus, $t\in\bigcup_{p\in S'}\semg{p}$, which leads to a
  contradiction.  We conclude that $\bigcup_{p\in S'}\semg{p} =
  \bigcup_{p\in S}\semg{p}$.
\end{proof}

\begin{lemma}[Subsume] 
\label{lemma:subsume}   
  Given two
  patterns $p,p'\in\TCX$ and a set of constructor patterns~$S$, if $p$
  is subsumed by $\{p'\}\cup S$ and $p'$ is subsumed by~$S$ then $p$ is
  subsumed by~$S$.
\end{lemma}
\begin{proof}
  We have $\semg{p}\subseteq\semg{p'}\cup\bigcup_{q\in S}\semg{q}$ and
  $\semg{p'}\subseteq\bigcup_{q\in S}\semg{q}$. This implies
  $\semg{p}\subseteq \bigcup_{q\in S}\semg{q}$.
\end{proof}

\minimize*
\begin{proof}
  We consider $\Smin=minimum(S)$.  We can first prove by induction on
  the size of $S$ that $\forall p\in S$, $p$ is subsumed by $\Smin$. The
  basic case is obvious. The interesting case is when a pattern $q\in
  S$ is not included in the kernel and thus in the final result. This
  could occur only if $q$ is subsumed by the kernel and in this case
  we can apply the induction hypothesis and Lemma~\ref{lemma:subsume}
  and obtain that $q$ is subsumed by $\Smin$.
  By construction, $\Smin$ is clearly included in $S$.  Consequently,
  by Lemma~\ref{lemma:validsubset} $\Smin$ is a valid subset of $S$.

  We can also prove by induction on the size of $S$ that for any $q$
  subsumed by $S$ we have
  $\card{minimum(S\cup\{q\})}\leq\card{minimum(S)}$. For this
  we proceed as before and we apply the induction hypothesis and
  Lemma~\ref{lemma:subsume}. On the other hand, by
  Lemma~\ref{lemma:validsubset}, if $S'$ is a valid subset of $S$ then
  all $q\in S\setminus S'$ are subsumed by~$S'$. We obtain that 
  $\card{minimum(S)}\leq\card{minimum(S')}$ and since 
  that $minimum(S)$ is a valid subset, it is also the smallest one.
\end{proof}



\convergenceAt*
\begin{proof}
  Recall that a meta-encoding of a complete approximation of the rule
  schema $\RrCat$ is provided in
  Appendix~\ref{ap:metaEncoding}. Automatic termination proof tools
  such as {\TTT} and {\aprove} have been used to prove that this
  meta-encoding is teminating and we can thus directly conclude to the
  termination of $\RrCat$.

  We can also extend the recursive path order (RPO) with status used
  in Lemma~\ref{lemma:convergence}.
  Once again the reduction ordering $\sizeorder$ is decreasing on all
  rules and strictly decreasing on $\rmMinusVMp$. The RPO $\rpo$ is
  strictly decreasing on all rules but $\rmMinusVMp$. The definitions
  are extended for as-patterns as follows:\\
  $
  \begin{array}{l@{\hspace{10pt}}l@{\hspace{10pt}}l@{\hspace{10pt}}l}
    \sizep(X\at{u}) = \sizep(u) \\
  \end{array}
  $\\
  with\\
  $
  \begin{array}{l@{\hspace{10pt}}l@{\hspace{10pt}}l@{\hspace{10pt}}l}
    \varp(X\at{v}) = \varp(v) & \\
    \symp(X\at{v}) = \symp({v}) &
  \end{array}
  $\\
  Note that we still have $\sizep(v)=\emptyset$ for any {\addt} term $v$ ({\ie}
  containing no $\minus$).  

  For the RPO order $\rpo$ we complete the quasi-precedence used in
  Lemma~\ref{lemma:convergence} to take into account aliases: 
  $\plus,\bot,f_n,\patvar,\at\rpo\minus$ for all $f_n\in\FF^{n}$ with
  the same statuses as before and with $\at:lexicographic ([1,2])$.
  It is easy to check that for each rule ${l}\raM{r}$ in $\RrCat$
  we have $l>r$.

  \medskip


  For the shape of the normal forms {\wrt} this system we proceed
  as in Lemma~\ref{lemma:convergence} and we suppose that a term of
  the form $u=u'\minus u''$ with $u',u''$ additive terms is a normal
  form.  For the same reasons as before, $u''$ cannot be a variable,
  $\bot$, a term whose head symbol is a $\plus$.  Moreover $u''$
  cannot be a term whose head symbol is an $\at$ since $u$ would be a
  redex for $\rmMinusMA$.
  If $u''$ is a term of the form $f(u_1,\ldots,u_n)$ then, as before,
  $u'$ cannot be a variable (since $u$ would be a redex for
  $\rmMinusVMp$ this time), a $\bot$, a term headed by a $\plus$, a
  term of the form $h(u'_1,\ldots,u'_m)$. Moreover $u'$ cannot be an
  aliased term (headed by $\at$) since $u$ would be a redex for
  $\rmMinusAM$.  Thus the normal forms contain no complements.

  We suppose that a term of the form $u=X\at u'$ with $u'$ a
  non-variable additive term is a normal form. Then $u'$ cannot be a
  $\bot$ since $u$ would be a redex for $\rmEmptyA$, cannot be a term
  whose head symbol is a $\plus$ since $u$ would be a redex for
  $\rmDistAt$. Thus normal forms could contain aliases only for
  constructor patterns.

  We consider now a normal form $u$ containing a $\bot$ at a position other than
  the root one. As before, $u$ 
  cannot be of the form $\bot\plus u'$,
  cannot be of the form $u'\plus\bot$,
  cannot be of the form $f(u_1,\ldots,\bot{},\ldots,$ $u_n)$,
  cannot be of the form $\bot\minus u'$ or $u'\minus\bot$ or $X \at \bot{}$ as seen above.
  Thus, normal form can contain $\bot$ only at the root position.
  
  \medskip

  For the confluence we first note that for any {\addt} pattern $v$,
  $\bot\minus v \multievalM{\RrCat}\bot$; the proof is easy by cases.
  We show then the local confluence of the system by proving that all
  critical pairs induced by the rewrite rules of the system converge.
  We have the same critical pairs as for Lemma~\ref{lemma:convergence}
  togheter with the following one and all converge:\\
  $(\rmAddL)-(\rmDistAt)$ and $(\rmAddR)-(\rmDistAt)$ (converge with $\rmEmptyA$ and $\rmEmptyA$, $(\rmAddL)$/$(\rmAddR)$),\\
  $(\rmEmptyA)-(\rmMinusAM)$  (converges with property of $\bot\minus v$ and  $\rmEmptyA$),\\
  $(\rmEmptyA)-(\rmMinusMA)$  (converges directly),\\
  $(\rmDistAdd)-(\rmMinusVMp)$ (converges with $\rmMinusMP$, $\rmMinusVMp$, $\rmMinusAM$ and $\rmDistAdd$, $\rmMinusMP$),\\
  $(\rmDistAt)-(\rmMinusAM)$ (converges with $\rmMinusPM$, and $\rmMinusPM$, $\rmDistAt$),\\
  $(\rmDistAt)-(\rmMinusMA)$ (converges with $\rmMinusMP$, twice $\rmMinusMA$ and $\rmMinusMP$),\\
  $(\rmMinusMV)-(\rmMinusAM)$ (converges with $\rmMinusMV$, $\rmEmptyA$),\\
  $(\rmMinusMB)-(\rmMinusAM)$ (converges with $\rmMinusMB$),\\
  $(\rmMinusMP)-(\rmMinusAM)$ (converges with twice $\rmMinusAM$, $\rmMinusMP$),\\
  $(\rmMinusAM)-(\rmMinusMA)$ (converges with twice $\rmMinusMA$, $\rmMinusAM$).

\end{proof}

\comEqAs*
\begin{proof}
  We proceed as in the proof of Proposition~\ref{prop:comEq} and we
  prove that the ground semantics of the left-hand side and right-hand
  side of the rewrite rules of $\RrCat$ are the same. The proof is the
  same for the rules already in $\RrC$.  For the other rules we use
  the fact that $\semg{\varVar}=\TC$ and thus, that for any pattern
  $p$, $\semg{p}\cap\semg{\varVar}=\semg{p}$.
\end{proof}


\mvarPreserve*
\begin{proof}
  By induction on the structure of patterns. If the reduction takes
  place at the root position then the property can be easily verified
  for each rewrite rule using the definition of matchable variables
  and basic laws of set theory. In fact, $\mvar{p}=\mvar{p'}$ for all
  rules except for $\rmMinusVMp$ for which
  $\mvar{p}\subseteq\mvar{p'}$.
  Moreover, $\var{p'}\subseteq\var{p}$ except for $\rmMinusVMp$ for which
  $\var{p'}\setminus\var{p}\subseteq\mvar{p'}$ and thus,
  $\fvar{p'}=\fvar{p}$ for this rule.
  We proceed similarly in the case where the reduction takes place at
  a position different from the root one. Thus, we check that the
  property is verified when $p$ is of the form $f(p_1,\ldots,p_n)$,
  ${p_1}\plus{p_2}$, ${p_1}\minus{p_2}$ and ${p_1}\at{p_2}$.
  If $f(p_1,\ldots,p_n)\evalM_{\RrCat}f(p_1',\ldots,p_n)$,
  since $p_1\evalM_{\RrCat}p_1'$ we have by induction
  $\mvar{p_1}\subseteq\mvar{p_1'}$ and consequently 
  $\mvar{p_1}\cup\ldots\cup\mvar{p_n}\subseteq\mvar{p_1'}\cup\ldots\cup\mvar{p_n}$
  which allows us to conclude for this case. 
  We can proceed similarly when $p={p_1}\at{p_2}$.
  The case ${p_1}\minus{p_2}\evalM_{\RrCat}{p_1}\minus{p_2'}$
  is trivial and for
  ${p_1}\minus{p_2}\evalM_{\RrCat}{p_1'}\minus{p_2}$ we can
  simply apply induction.
  For the case
  ${p_1}\plus{p_2}\evalM_{\RrCat}{p_1'}\plus{p_2}$ we can
  apply induction and use the fact that if
  $\mvar{p_1}\subseteq\mvar{p_1'}$ then
  $\mvar{p_1}\cap\mvar{p_2}\subseteq\mvar{p_1'}\cap\mvar{p_2}$.
\end{proof}

\begin{lemma}[Semantics preservation under substitution application]
  \label{lemma:semgStableBySubst}
  Given the patterns $p,p'$ such that $p\evalM_{\RrCat}p'$
  and a substitution $\sigma$ with
  $\dom{\sigma}\cap\fvar{\ovr{p}}=\emptyset$, we have
  $\semg{\sigma(p)}=\semg{\sigma(p')}$.
\end{lemma}
\begin{proof}
  If the reduction takes place at the root position then the property
  can be verified for each rewrite rule.

  If $p=v_1\plus v_2$ then one of the rules $(\rmAddL)$ or
  $(\rmAddR)$ is used for the reduction at the root position. 
  If $p=\bot\plus{q}$ then $p'=q$ with
  $\sigma(p)=\bot\plus\sigma({q})$ and $\sigma(p')=\sigma({q})$.  We
  have
  $\semg{\sigma(p)}=\semg{\bot\plus\sigma({q})}=\semg{\bot}\cup\semg{\sigma({q})}=\emptyset\cup\semg{\sigma({q})}=\semg{\sigma({q})}=\semg{\sigma(p')}$.
  The case $p={q}\plus\bot$ is similar.

  If $p=f(v_1,\ldots,v_n)$ then one of the rules $(\rmEmptyF)$ or
  $(\rmDistAdd)$ is used for the reduction at the root position. 
  In the former case $p=f({v_1},\ldots,\bot{},\ldots,{v_n})$ and
  $\sigma(p)=f(\sigma(v_1),\ldots,\sigma(\bot{}),\ldots,\sigma(v_n))=f(\sigma(v_1),
  \ldots, \bot{}, \ldots,$ $\sigma(v_n))$.
  Since patterns are linear we can apply the recursive definition for
  constructors:
  $\semg{f(\sigma(v_1),\ldots,\bot{},$ $\ldots,$ $\sigma(v_n))}=\{f(u_1,\ldots,u_n)\mid
  (u_1,\ldots,u_n)\in\semg{\sigma(v_1)}\times\ldots\times\semg{\sigma(v_n)}\}$.
  Since the semantics of $\bot{}$ is an empty set then so is the
  semantics of $f(\sigma(v_1),\ldots,\bot{},\ldots,\sigma(v_n)))$.
  We also have $\sigma(p')=\sigma(\bot{})=\bot{}$ and since
  $\semg{\bot{}}$ is the empty set, the property is verified for rule
  $(\rmEmptyF)$.
  For the case where the rule $(\rmDistAdd)$ is applied we have
  $p=f({v_1},\ldots,{v_i'}\plus{v_i''},\ldots,{v_n})$ with
  $\sigma(p)=f(\sigma({v_1}),\ldots,\sigma({v_i'}\plus{v_i''}),\ldots,\sigma({v_n}))=f(\sigma({v_1}),\ldots,\sigma({v_i'})\plus\sigma({v_i''}),\ldots,\sigma({v_n}))$
  and $p'=f({v_1},\ldots,{v_i'},\ldots,{v_n})\plus
  f({v_1},\ldots,{v_i''},\ldots,{v_n})$ with
  $\sigma(p')=f(\sigma({v_1}),$ $\ldots,\sigma({v_i'}),\ldots,\sigma({v_n}))\plus
  f(\sigma({v_1}),\ldots,\sigma({v_i''}),\ldots,\sigma({v_n}))$.
  We have
  $\semg{\sigma(p)}=\semg{f(\sigma({v_1}),\ldots,\sigma({v_i'})\plus\sigma({v_i''}),\ldots,\sigma({v_n}))}=$
  $\{f(u_1,\ldots,$ $u_n)\mid (u_1,\ldots,u_n)\in\semg{\sigma(v_1)}\times\ldots\times\semg{\sigma({v_i'})\plus\sigma({v_i''})}\ldots\times\semg{\sigma(v_n)}\}=
  \{f(u_1,\ldots,u_n)\mid  (u_1,\ldots,u_n)\in\semg{\sigma(v_1)}\times\ldots\times\semg{\sigma({v_i'})}\cup\semg{\sigma({v_i''})}\times\ldots$   $\times\semg{\sigma(v_n)}\}=
  \{f(u_1,\ldots,u_n)\mid (u_1,\ldots,u_n)\in\semg{\sigma(v_1)}\times\ldots\times\semg{\sigma({v_i'})}\ldots\times\semg{\sigma(v_n)}\}
  \cup\{f(u_1,\ldots,u_n)\mid
  (u_1,\ldots,u_n)\in\semg{\sigma(v_1)}\times\ldots\times\semg{\sigma({v_i''})}\ldots\times\semg{\sigma(v_n)}\}=
  \semg{f(\sigma(v_1),\ldots,\sigma({v_i'}),\ldots,$ $\sigma(v_n))}\cup \semg{f(\sigma(v_1),\ldots,\sigma({v_i''}),\ldots,\sigma(v_n))}=\semg{\sigma(p')}$
  and thus the property is verified for rule
  $(\rmDistAdd)$ as well.

  If $p=x\at{p}$ with $x\in\XX$  then one of the rules $(\rmEmptyA)$ or
  $(\rmDistAt)$ is used for the reduction at the root position. 
  If $p=x\at{\bot}$ then the rule $(\rmEmptyA)$ is applied
  and $p'=\bot$.  We have
  $\semg{\sigma(p)}=
  \semg{\sigma(x)\at\sigma(\bot)}=
  \semg{\sigma(x)\at\bot}=
  \semg{\sigma(x)}\cap\semg{\bot}=
  \semg{\sigma(x)}\cap\emptyset=
  \emptyset=
  \emptyset=
  \semg{\bot}=
  \semg{\sigma(\bot)}=
  \semg{\sigma(p')}$.
  If $p=x\at({v_1}\plus{v_2})$ then the rule $(\rmDistAt)$ is applied
  and $p'=x\at{v_1}\plus{x}\at{v_2}$.  We have
  $\semg{\sigma(p)}=\semg{\sigma(x)\at\sigma({v_1}\plus{v_2})}=
  \semg{\sigma(x)\at(\sigma({v_1})\plus\sigma({v_2}))}=
  \semg{\sigma(x)}\cap(\semg{\sigma({v_1})}\cup\semg{\sigma({v_2})})=
  (\semg{\sigma(x)}\cap\semg{\sigma({v_1})})\cup(\semg{\sigma(x)}\cap\semg{\sigma({v_2})})=
  \semg{\sigma(x)\at{\sigma({v_1})}}\cup\semg{\sigma(x)\at{\sigma({v_2})}}=
  \semg{(\sigma(x)\at{\sigma({v_1})})\plus({\sigma(x)\at{\sigma({v_2})}})}=
  \semg{\sigma(p')}$.

  When $p=p_1\minus{p_2}$ one of the other rules applies. Since
  $\dom{\sigma}\subseteq\mvar{p}=\mvar{p_1}$ we have
  $\sigma(p)=\sigma(p_1)\minus{p_2}$. We apply the definition of
  ground semantics of extended terms and several identifies of the
  algebra of sets.
  
  If $p=v\minus{x}$ with $x\in\XX$ then the rule $(\rmMinusMV)$ is applied
  and $p'=\bot$.
  We have $\sigma(p)=\sigma(v)\minus{x}$ and
  $\semg{\sigma(p)}=\semg{\sigma(v)\minus{x}}=\semg{\sigma(v)}\minus\semg{{x}}=\semg{\sigma(v)}\minus\TC=\emptyset=\semg{\bot}=\semg{\sigma(p')}$.

  If $p=v\minus\bot$ then the rule $(\rmMinusMB)$ is applied and
  $p'=v$. We have
  $\semg{\sigma(p)}=\semg{\sigma(v)\minus\bot}=\semg{\sigma(v)}\setminus\semg{\bot}=\semg{\sigma(v)}\setminus\emptyset=\semg{\sigma(v)}=\semg{\sigma(p')}$.

  If $p=w\minus({v_1}\plus{v_2})$ then the rule $(\rmMinusMP)$ is applied
  and $p'=(w\minus{v_1})\minus{v_2}$. We have
  $\semg{\sigma(p)}=\semg{\sigma(w) \minus (v_1 \plus v_2)}=
  \semg{\sigma(w)} \setminus \semg{(v_1 \plus v_2)}=
  \semg{\sigma(w)} \setminus (\semg{v_1} \cup \semg{v_2})$
  and
  $\semg{\sigma(p')}=\semg{(\sigma(w) \minus v_1) \minus v_2}=
  \semg{(\sigma(w) \minus v_1)} \setminus \semg{v_2}=
  (\semg{\sigma(w)} \setminus \semg{v_1}) \setminus \semg{v_2}=
  \semg{\sigma(w)} \setminus (\semg{v_1} \cup \semg{v_2})$. Thus, 
  $\semg{\sigma(p)}=\semg{\sigma(p')}$.

  If $p={x}\minus{g(\meta{t_1},\ldots,\meta{t_n})}$ with $x\in\XX$ then the rule $(\rmMinusVMp)$ is
  applied and $p'=x\at(\sum\limits_{c\in\CC} c(\varCt_1,\ldots,\varCt_{m}) \minus g(\meta{t_1},\ldots,\meta{t_n}))$.
  Since $\dom{\sigma}\subseteq\mvar{p}=x$ and since
  $\varCt_1,\ldots,\varCt_{m}$ are fresh variables different from $x$
  we have
  $\sigma(p)=\sigma({x})\minus{g(\meta{t_1},\ldots,\meta{t_n})}$ and 
  $\sigma(p')=\sigma(x)\at(\sum\limits_{c\in\CC}c(\varCt_1,\ldots,\varCt_{m})\minus{g(\meta{t_1},\ldots,\meta{t_n})})$.
  We have
  $\semg{\sigma(p)}=\semg{\sigma({x})\minus{g(\meta{t_1},\ldots,\meta{t_n})}}=\semg{\sigma({x})}\setminus\semg{{g(\meta{t_1},\ldots,\meta{t_n})}}$
  and 
  $\semg{\sigma(p')}=\semg{\sigma(x)\at(\sum\limits_{c\in\CC}c(\varCt_1,\ldots,\varCt_{m})\minus{g(\meta{t_1},\ldots,\meta{t_n})})}=
  \semg{\sigma(x)}\cap\semg{\sum\limits_{c\in\CC}c(\varCt_1,\ldots,\varCt_{m})\minus{g(\meta{t_1},\ldots,\meta{t_n})}}=
  \semg{\sigma(x)}\cap(\semg{\sum\limits_{c\in\CC}c(\varCt_1,\ldots,\varCt_{m})}\setminus\semg{{g(\meta{t_1},\ldots,\meta{t_n})}}=
  \semg{\sigma(x)}\cap(\TC\setminus\semg{{g(\meta{t_1},\ldots,\meta{t_n})}}=
  (\semg{\sigma(x)}\cap\TC)\setminus(\semg{\sigma(x)}\cap\semg{{g(\meta{t_1},\ldots,\meta{t_n})}})=
  \semg{\sigma(x)}\setminus(\semg{\sigma(x)}\cap\semg{{g(\meta{t_1},\ldots,\meta{t_n})}})=
  (\semg{\sigma(x)}\setminus\semg{\sigma(x)})\cup(\semg{\sigma(x)}\setminus\semg{{g(\meta{t_1},\ldots,\meta{t_n})}})=
  \semg{\sigma({x})}\setminus\semg{{g(\meta{t_1},\ldots,\meta{t_n})}}$
  and thus,   $\semg{\sigma(p)}=\semg{\sigma(p')}$.

  If $p=\bot \minus f({v_1},\ldots,{v_n})$ then the rule $(\rmMinusBM)$ is applied
  and $p'=\bot$. We have
  $\semg{\sigma(p)}=\semg{\bot\minus{f({v_1},\ldots,{v_n})}}=
  \semg{\bot}\setminus\semg{f({v_1},\ldots,{v_n})}=
  \emptyset=
  \semg{\bot}=\semg{\sigma(p')}
  $

  If $p=({\vm} \plus {\wm}) \minus f({v_1},\ldots,{v_n})$ then the rule $(\rmMinusPM)$ is applied
  and $p'=({\vm} \minus f({v_1},\ldots,{v_n})) \plus ({\wm} \minus f({v_1},\ldots,{v_n}))$. We have
  $\semg{\sigma(p)}=\semg{\sigma({\vm}\plus{\wm})\minus{f({v_1},\ldots,{v_n})}}=
  \semg{(\sigma({\vm})\plus\sigma({\wm}))\minus{f({v_1},\ldots,{v_n})}}=
  \semg{\sigma({\vm})\plus\sigma({\wm})}\setminus\semg{f({v_1},\ldots,{v_n})}=
  (\semg{\sigma({\vm})}\cup\semg{\sigma({\wm})})\setminus\semg{f({v_1},\ldots,{v_n})}=
  \semg{\sigma({\vm})}\setminus\semg{f({v_1},\ldots,{v_n})}\cup\semg{\sigma({\wm})}\setminus\semg{f({v_1},\ldots,{v_n})}=
  \semg{\sigma({\vm})\minus{f({v_1},\ldots,{v_n})}}\cup\semg{\sigma({\wm})\minus{f({v_1},\ldots,{v_n})}}=
  \semg{(\sigma({\vm})\minus{f({v_1},\ldots,{v_n})})\plus{(\sigma({\wm})\minus{f({v_1},\ldots,{v_n})})}}=
  \semg{\sigma(p')}
  $

  If $p=f({v_1},\ldots,{v_n}) \minus f({t_1},\ldots,{t_n})$ then the
  rule $(\rmMinusFF)$ is applied and $p'= f({v_1}\minus
  {t_1},\ldots,{v_n}) \plus \cdots \plus f({v_1},\ldots,{v_n}\minus
  {t_n})$.
  We have
  $\semg{\sigma(p)}=
  \semg{\sigma(f(\vm_1,\ldots,\vm_n))\minus{f(\tm_1,\ldots,\tm_n)}}=
  \semg{\sigma(f(\vm_1,\ldots,\vm_n))}\setminus\semg{f(\tm_1,\ldots,\tm_n)}=
  \semg{f(\sigma(\vm_1),\ldots,\sigma(\vm_n))}\setminus\semg{f(\tm_1,\ldots,\tm_n)}$
  and since patterns are linear we can use the recursive definition and we get 
  $\semg{\sigma(p)}=
  \{f(\vm'_1,\ldots,\vm'_n)\mid(\vm'_1,\ldots,\vm'_n)\in\semg{\sigma(\vm_1)}\times\ldots\times\semg{\sigma(\vm_n)}\}
  \setminus
  \{f(\tm'_1,\ldots,\tm'_n)\mid(\tm'_1,\ldots,\tm'_n)\in\semg{\tm_1}\times\ldots\times\semg{\tm_n}\}$
  and similarly
  $\semg{\sigma(p')}=
  \semg{\sigma(f(\vm_1\minus \tm_1,\ldots,\vm_n))\plus\ldots\plus\sigma(f(\vm_1,\ldots,\vm_n\minus \tm_n))}=
  \semg{\sigma(f(\vm_1\minus \tm_1,\ldots,\vm_n))}\cup\ldots\cup\semg{\sigma(f(\vm_1,\ldots,\vm_n\minus \tm_n))}=
  \semg{f(\sigma(\vm_1)\minus \tm_1,\ldots,\sigma(\vm_n)))}\cup\ldots\cup\semg{f(\sigma(\vm_1),\ldots,\sigma(\vm_n)\minus \tm_n)}=
  \{f(\vm'_1,\ldots,\vm'_n)\mid(\vm'_1,\ldots,\vm'_n)\in\semg{\sigma(\vm_1)\minus\tm_1}\times\ldots\times\semg{\sigma(\vm_n)}\}
  \cup \ldots \cup
  \{f(\vm'_1,\ldots,\vm'_n)\mid(\vm'_1,\ldots,\vm'_n)\in\semg{\sigma(\vm_1)}\times\ldots\times\semg{\sigma(\vm_n)\minus \tm_n}\}$.
  We have to show that  $\forall s, s\in\semg{\sigma(p)}  \text{ iff } s\in\semg{\sigma(p')}$.
  Take $s\in\semg{\sigma(p)}$, then 
  $s\in\{f(\vm'_1,\ldots,\vm'_n)\mid(\vm'_1,\ldots,\vm'_n)\in\semg{\sigma(\vm_1)}\times\ldots\times\semg{\sigma(\vm_n)}\}$ 
  and 
  $s\not\in\{f(\tm'_1,\ldots,\tm'_n)\mid(\tm'_1,\ldots,\tm'_n)\in\semg{\tm_1}\times\ldots\times\semg{\tm_n}\}$  
  and thus, 
  $s=f(\wm_1,\ldots,\wm_n)$ s.t. $\forall i,
  \wm_i\in\semg{\sigma(\vm_i)}$ and $\exists j,
  \wm_j\not\in\semg{\tm_j}$. Consequently,
  $\wm_j\in\semg{\sigma(\vm_j)}\setminus\semg{\tm_j}$ and thus,
  $f(\wm_1,\ldots,\wm_n)\in\{f(\vm'_1,\ldots,\vm'_n)\mid(\vm'_1,\ldots,\vm'_n)\in\semg{\vm_1}\times\ldots\times\semg{\sigma(\vm_j)\minus\tm_j}\times\ldots\times\semg{\vm_n}\}$
  $\subseteq\semg{\sigma(p')}$. 
  We can show similarly that if $s\in\semg{\sigma(p')}$ then $s\in\semg{\sigma(p)}$.

  If $p=f({v_1},\ldots,{v_n}) \minus g({t_1},\ldots,{t_n})$ with
  $f\not=g$, then the rule $(\rmMinusFG)$ is applied and
  $p'=f({v_1},\ldots,{v_n})$.  We have that
  $\semg{\sigma(p)}=
  \semg{\sigma(f(\vm_1,\ldots,\vm_n))\minus{g(\tm_1,\ldots,\tm_n)}}=
  \semg{\sigma(f(\vm_1,\ldots,\vm_n))}\setminus\semg{g(\tm_1,\ldots,$ $\tm_n)}=
  \semg{f(\sigma(\vm_1),\ldots,\sigma(\vm_n))}\setminus\semg{g(\tm_1,\ldots,\tm_n)}$
  and since patterns are linear we can use the recursive definition and we get 
  $\semg{\sigma(p)}=
  \{f(\vm'_1,\ldots,\vm'_n)\mid(\vm'_1,\ldots,\vm'_n)\in\semg{\sigma(\vm_1)}\times\ldots\times\semg{\sigma(\vm_n)}\}
  \setminus
  \{g(\tm'_1,\ldots,$ $\tm'_n)\mid(\tm'_1,\ldots,\tm'_n)\in\semg{\tm_1}\times\ldots\times\semg{\tm_n}\}=
  \{f(\vm'_1,\ldots,\vm'_n)\mid(\vm'_1,\ldots,\vm'_n)\in\semg{\sigma(\vm_1)}\times\ldots\times\semg{\sigma(\vm_n)}\}=
  \semg{\sigma(p')}$.

  If $p=x\at{\vm}\minus{\wm}$ with $x\in\XX$ then the rule $(\rmMinusAM)$ is applied and
  $p'=x\at({\vm}\minus{\wm})$.  We have $\semg{\sigma(p)}=
  \semg{\sigma(x)\at\sigma({\vm})\minus{\wm}}=
  (\semg{\sigma(x)}\cap\semg{\sigma({\vm})})\setminus\semg{\wm}=
  \semg{\sigma(x)}\cap(\semg{\sigma({\vm})}\setminus\semg{\wm})=
  \semg{\sigma(x)\at(\sigma({\vm})\minus{\wm})}=
  \semg{\sigma(p')}
  $ (we used the fact that intersection with set difference is set difference with intersection).

  If $p={\vm}\minus{x}\at{\wm}$ with $x\in\XX$ then the rule $(\rmMinusMA)$ is applied and
  $p'={\vm}\minus{\wm}$.  We have $\semg{\sigma(p)}=
  \semg{\sigma({\vm})\minus{x}\at{\wm}}=
  \semg{\sigma({\vm})}\setminus(\semg{x}\cap\semg{\wm})=
  \semg{\sigma({\vm})}\setminus(\TC\cap\semg{\wm})=
  \semg{\sigma({\vm})}\setminus\semg{\wm}=
  \semg{\sigma({\vm})\minus{\wm}}=
  \semg{\sigma(p')}
  $.

  We proceed similarly in the case where the reduction takes place at
  a position different from the root one. Thus, we check that the
  property is verified when $p$ is of the form $f(p_1,\ldots,p_n)$,
  ${p_1}\plus{p_2}$, ${p_1}\minus{p_2}$ and ${p_1}\at{p_2}$.
  First, it is easy to check that
  $\fvar{p_i}\subseteq\fvar{p_1,\ldots,p_n}$ for all
  $i\in[1,\ldots,n]$ and
  $\fvar{p_1},\fvar{p_2}\subseteq\fvar{{p_1}\plus{p_2}},\fvar{{p_1}\minus{p_2}},\fvar{{p_1}\at{p_2}}$;
  consequently,\linebreak
 $\dom{\sigma}\cap\fvar{\ovr{p_i}}=\emptyset$, for any
  $i\in[1,\ldots,n]$.

  If $p=f(p_1,\ldots,p_i,\ldots,p_n)$ and
  $p'=f(p_1,\ldots,p_i',\ldots,p_n)$ with
  $p_i\evalM_{\Rr}p_i'$ for some $i\in[1,\ldots,n]$ we have,
  by induction, $\semg{\sigma(p_i)}=\semg{\sigma(p_i')}$.
  Then,
  $\semg{\sigma(p)}=
  \semg{f(\sigma(p_1),\ldots,\sigma(p_i),\ldots,\sigma(p_n))}=
  \{f(u_1,$ $\ldots,u_n)\mid(u_1,\ldots,u_n)\in\semg{\sigma(p_1)}\times\ldots\times\semg{\sigma({p_i})}\times\ldots\times\semg{\sigma(p_n)}\}=
  \{f(u_1,\ldots,u_n)\mid(u_1,\ldots,u_n)\in\semg{\sigma(p_1)}\times\ldots\times\semg{\sigma({p_i'})}\times\ldots\times\semg{\sigma(p_n)}\}=
  \semg{f(\sigma(p_1),\ldots,\sigma(p_i'),\ldots,\sigma(p_n))}=
  \semg{\sigma(p')}$.

  If $p={p_1}\plus{p_2}$ and
  $p'={p_1'}\plus{p_2}$ with
  $p_1\evalM_{\Rr}p_1'$ we have,
  by induction, $\semg{\sigma(p_1)}=\semg{\sigma(p_1')}$.
  Then,
  $\semg{\sigma(p)}=
  \semg{\sigma({p_1})\plus\sigma({p_2})}=
  \semg{\sigma({p_1})}\cup\semg{\sigma({p_2})}=
  \semg{\sigma({p_1'})}\cup\semg{\sigma({p_2})}=
  \semg{\sigma({p_1}')\plus\sigma({p_2})}=
  \semg{\sigma(p')}$.
  We can proceed similarly when
  $p_2\evalM_{\Rr}p_2'$ and when 
  If $p={p_1}\at{p_2}$.

  If $p={p_1}\minus{p_2}$ and
  $p'={p_1'}\at{p_2}$ with
  $p_1\evalM_{\Rr}p_1'$ we have,
  by induction, $\semg{\sigma(p_1)}=\semg{\sigma(p_1')}$.
  Then,
  $\semg{\sigma(p)}=
  \semg{\sigma({p_1})\at{p_2}}=
  \semg{\sigma({p_1})}\cap\semg{{p_2}}=
  \semg{\sigma({p_1'})}\cap\semg{{p_2}}=
  \semg{\sigma({p_1}')\at\sigma({p_2})}=
  \semg{\sigma(p')}$.
  The case where $p_2\evalM_{\Rr}p_2'$ is trivial.
\end{proof}

\compEnc*
\begin{proof}
  We consider the list of rules
  $\LL=[\foo(\ovr{p_1})\ra{t_1},\ldots,$ $\foo(\ovr{p_n})\ra{t_n}]$ and
  thus, we have $\TRcomp(\LL)=\oplus_{k=1}^{n} [\foo(\ovr{q_1^k})\ra t_k
  \otrsep \ldots \otrsep$ $\foo(\ovr{q_m^k})\ra t_k \mid
  \ovr{q_1^k}\plus\ldots\plus\ovr{q_m^k}=\ovr{p_k}\downarrow_{\RrCat}\neq\bot,
  \ovr{q_1^k},\ldots,\ovr{q_m^k} \text{ contain} \text{ no}$ $\text{ symbol }
  \plus]$.
  If $t\eval_{\LL}t'$ then there exist $i\in[1,\ldots,n]$,
  $\omega\in\PPos(t)$ and a substitution $\sigma$ such that $\stt t
  \omega = \foo(\ovr{v})$, $\dom{\sigma}=\mvar{\ovr{p_i}}$, $\ovr{v}
  \in \semg{\sigma(\ovr{p_i})}$, $\ovr{p_j} \nmatcha \ovr{v}, \forall j<i$ and
  $t'=\rmp t \omega {\sigma(t_i)}$.  Consequently,
  $\ovr{v}\not\in\semg{p_j}, \forall j<i$.
  According to Lemma~\ref{lemma:semgStableBySubst}
  $\semg{\tau(\ovr{p_k})}=\semg{\tau(\ovr{q_1^k}\plus\ldots\plus\ovr{q_m^k})}$
  for all $k\in[1,\ldots,n]$ and all substitution $\tau$
  s.t. $\dom{\tau}\cap\fvar{\ovr{p_k}}=\emptyset$ and, in particular,
  for $\tau$  the identity substitution. 
  Since $\ovr{v}\in\semg{\sigma(\ovr{p_i})}$
  then
  $\ovr{v}\in\semg{\sigma(\ovr{q_1^i}\plus\ldots\plus\ovr{q_m^i})}=$
  $\semg{\sigma(\ovr{q_1^i})\plus\ldots\plus\sigma(\ovr{q_m^i})}=$
  $\semg{\sigma(\ovr{q_1^i})}\cup\ldots\cup\semg{\sigma(\ovr{q_m^i})}$.
  Since  $\ovr{v}\not\in\semg{p_j}, \forall j<i$
  then
  $\ovr{v}\not\in\semg{\ovr{q_1^j}\plus\ldots\plus\ovr{q_m^j}}=\semg{\ovr{q_1^j}}\cup\ldots\cup\semg{\ovr{q_m^j}},
  \forall j<i$
  and thus 
  $\ovr{q_k^j}\nmatcha{v}, \forall{j<i},k\in[1,\ldots,m]$. 
  Consequently, one of the rules $\foo(\ovr{q_k^i})\ra t_i$,
  $k\in[1,\ldots,m]$ in $\TRcomp(\LL)$ apply and
  $t\eval_{\TRcomp(\LL)}\rmp t \omega {\sigma(t_i)}$.
  We proceed similarly when considering $t\eval_{\TRcomp(\LL)}t'$.
\end{proof}

\aliasEncoding*
\begin{proof}
  We consider the rule $e:\foo(\rmp{p}{\omega}{x\at q})\ra r$ in $\LL$, with
  $\omega$ such that $\forall\omega'\in\PPos(q),q(\omega')\not=\at$, is
  applied to reduce the term $t$ into $t'$. For simplicity we suppose
  $r$ contains exactly one occurence of $x$, the case where $x$
  doesn't occur or occurs several times in $r$ are handled in exactly
  the same way. 
  We show that the reduction is preserved at each (intermediate) step
  of the transformation: $t\eval_{e}t'$ iff
  $t\eval_{\TRat(e)}t'$ with
  $\TRat(e)=\TRat(\rmp{p}{\omega}{x\at{q}}\ra\rmp{r}{\psi}{x})=\rmp{p}{\omega}{q}\ra\rmp{r}{\psi}{q}$.

  Since $t\eval_{\LL} t'$ using the rule
  $\rmp{p}{\omega}{x\at{q}}\ra\rmp{r}{\psi}{x}$ then there exist
  $\omega\in\PPos(t)$ and a substitution $\sigma$ such that $\stt t
  \omega = \foo({v})$,
  $\dom{\sigma}=\mvar{\rmp{p}{\omega}{x\at{q}}}$,
  ${v}\in\semg{\sigma(\rmp{p}{\omega}{x\at{q}})}$ and $t'=\rmp{t}{\omega}{\sigma(\rmp{r}{\psi}{x})}$. 
  Since $\dom{\sigma}=\mvar{\rmp{p}{\omega}{x\at{q}}}$ and $p$ is a
  constructor pattern then $\semg{\sigma(\rmp{p}{\omega}{x\at{q}})}$
  is a singleton. Moreover
  ${\sigma(\rmp{p}{\omega}{x\at{q}})}={\rmp{\sigma(p)}{\omega}{\sigma(x\at{q})}}={\rmp{\sigma(p)}{\omega}{\sigma(x)\at{\sigma(q)}}}$.
  It is easy to check that since this latter pattern is ground we have
  $\semg{\rmp{\sigma(p)}{\omega}{\sigma(x)\at{\sigma(q)}}}=$
  $\semg{\rmp{\sigma(p)}{\omega}{\sigma(x)}}\cap\semg{\rmp{\sigma(p)}{\omega}{\sigma(q)}}$
  and consequently that we have
  $\semg{\rmp{\sigma(p)}{\omega}{\sigma(x)\at{\sigma(q)}}}=$
  $\semg{\rmp{\sigma(p)}{\omega}{\sigma(x)}}=$ 
  $\semg{\rmp{\sigma(p)}{\omega}{\sigma(q)}}$. Since all these
  semantics are singletons we have that
  $\rmp{\sigma(p)}{\omega}{\sigma(x)}=$
  $\rmp{\sigma(p)}{\omega}{\sigma(q)}$ and ${\sigma(x)}=$
  ${\sigma(q)}$.
  Since we have that $\semg{\sigma(\rmp{p}{\omega}{x\at{q}})}=$
  $\semg{\sigma(\rmp{p}{\omega}{q})}$ then
  ${v}\in\semg{\sigma(\rmp{p}{\omega}{q})}$ and thus $t$  reduces
  {\wrt} to the rule $\rmp{p}{\omega}{q}\ra\rmp{r}{\psi}{q}$ to 
  $\rmp{t}{\omega}{\sigma(\rmp{r}{\psi}{q})}=$ 
  $\rmp{t}{\omega}{\rmp{\sigma(r)}{\psi}{\sigma(q)})}=$
  $\rmp{t}{\omega}{\rmp{\sigma(r)}{\psi}{\sigma(x)})}=$. 
  $\rmp{t}{\omega}{\sigma(\rmp{r}{\psi}{x})}=t'$. 

  We can use the same reasoning for the other direction.
\end{proof}

\orderEnc*
\begin{proof}
  We consider the list of rules
  $\LL=[\foo(\ovr{p_1})\ra{t_1},\ldots,$ $\foo(\ovr{p_n})\ra{t_n}]$ and
  thus, we have $\TRord(\LL)=\cup_{k=1}^{n} 
  \{
  \foo(\ovr{q_1^k})\ra t_k \otrsep \ldots \otrsep$ $\foo(\ovr{q_m^k})\ra t_k \mid
  \ovr{q_1^k}\plus\ldots\plus\ovr{q_m^k}=\ovr{p_k}\minus(\ovr{p_1}\plus\ldots\plus \ovr{p_{k-1}})\downarrow_{\RrCat}\neq\bot,
  \ovr{q_1^k},\ldots,\ovr{q_m^k} \text{ contain no}\text{ symbol } \plus
  \}$.
  If $t\eval_{\LL}t'$ then there exist $i\in[1,\ldots,n]$,
  $\omega\in\PPos(t)$ and a substitution $\sigma$ such that $\stt t
  \omega = \foo(\ovr{v})$, $\dom{\sigma}=\mvar{\ovr{p_i}}$, $\ovr{v}
  \in \semg{\sigma(\ovr{p_i})}$, $\ovr{p_j} \nmatcha \ovr{v}, \forall j<i$
  and $t'=\rmp t \omega {\sigma(t_i)}$.  Consequently,
  $\ovr{v}\not\in\semg{p_j}, \forall j<i$ or equivalently
  $\ovr{v}\not\in\semg{\ovr{p_1}}\cup\ldots\cup\semg{\ovr{p_{i-1}}}$.

  According to Lemma~\ref{lemma:semgStableBySubst} we have
  $\semg{\tau(\ovr{p_k}\minus(\ovr{p_1}\plus\ldots\plus\ovr{p_{k-1}}))}=$
  $\semg{\tau(\ovr{q_1^k}\plus\ldots\plus\ovr{q_m^k})}$
  for all $k\in[1,\ldots,n]$ and all substitution $\tau$
  s.t. $\dom{\tau}\cap\fvar{\ovr{p_k}}=\emptyset$.
  We have
  $\dom{\sigma}=\mvar{\ovr{p_i}}$
  and thus
  $\dom{\sigma}\cap\fvar{\ovr{p_i}\minus(\ovr{p_1}\plus\ldots\plus\ovr{p_{i-1}})}=\emptyset$,
  and since
  $\semg{\sigma(\ovr{p_i}\minus(\ovr{p_1}\plus\ldots\plus\ovr{p_{i-1}}))}=$
  $\semg{\sigma(\ovr{p_i})\minus(\ovr{p_1}\plus\ldots\plus\ovr{p_{i-1}})}=$
  $\semg{\sigma(\ovr{p_i})}\setminus\semg{\ovr{p_1}\plus\ldots\plus\ovr{p_{i-1}}}=$
  $\semg{\sigma(\ovr{p_i})}\setminus(\semg{\ovr{p_1}}\cup\ldots\cup\semg{\ovr{p_{i-1}}})$
  and
  $\semg{\sigma(\ovr{q_1^i}\plus\ldots\plus\ovr{q_m^i})}=$
  $\semg{\sigma(\ovr{q_1^i})\plus\ldots\plus\sigma(\ovr{q_m^i})}=$
  $\semg{\sigma(\ovr{q_1^i})}\cup\ldots\cup\semg{\sigma(\ovr{q_m^i})}$
  then, 
  $\semg{\sigma(\ovr{p_i})}\setminus(\semg{\ovr{p_1}}\cup\ldots\cup\semg{\ovr{p_{i-1}}})=$
  $\semg{\sigma(\ovr{q_1^i})}\cup\ldots\cup\semg{\sigma(\ovr{q_m^i})}$.
  Thus
  $\ovr{v}\in\semg{\sigma(\ovr{p_i})}$ and
  $\ovr{v}\not\in\semg{\ovr{p_1}}\cup\ldots\cup\semg{\ovr{p_{i-1}}}$
  iff
  $\ovr{v}\in\semg{\sigma(\ovr{q_1^i})}\cup\ldots\cup\semg{\sigma(\ovr{q_m^i})}$.
  Consequently, one of the rules $\foo(\ovr{q_k^i})\ra t_i$,
  $k\in[1,\ldots,m]$ in $\TRcomp(\LL)$ apply and
  $t\eval_{\TRcomp(\LL)}\rmp t \omega {\sigma(t_i)}$.
  We proceed similarly when considering $t\eval_{\TRord(\LL)}t'$.
\end{proof}

\section{Meta encoding of the rewriting system $\RrCat$}
\label{ap:metaEncoding}

The meta encoding of the rule schemas in
Figure~\ref{fig:elimComplementAs} is given below in a syntax usable by
{\aprove}/{\TTT}. Both {\aprove} and {\TTT} can be used to prove the
termination of this rewriting system.

\begin{verbatim}
(VAR u u1 u2 v v1 v2 w f g lu lv n m i tail sig)
(RULES
  plus(bot,v) -> v

  plus(v,bot) -> v

  appl(f,lv) -> split(f,lv,nil)
  split(f,cons(u,lu),lv) -> split(f,lu,cons(u,lv))
  split(f,cons(bot,lu),lv) -> bot
  split(f,cons(plus(u1,u2),lu),lv) -> 
        plus(Appl(f,rest(lu,cons(u1,lv))),
             Appl(f,rest(lu,cons(u2,lv))))
  split(f,nil,lv) -> frozen(f,rest(nil,lv))
  rest(lu,nil) -> lu
  rest(lu,cons(u,lv)) -> rest(cons(u,lu),lv)

  at(var(n),bot) -> bot

  at(var(n),plus(u1,u2)) -> 
        plus(at(var(n),u1),at(var(n),u2))

  minus(v, var(n), sig) -> bot
  minus(v, bot, sig) -> v
  minus(w, plus(v1,v2), sig) -> 
        minus(minus(w,v1, sig),v2, sig)

  minus(var(m), appl(f,lv), sig) -> 
        at(var(m),minus(gensum(sig),appl(f,lv), sig))
  gensum(nilsig) -> bot
  gensum(conssig(f,n,tail)) -> 
        plus(appl(f,genvar(n)), gensum(tail))
  genvar(z) -> nil
  genvar(s(n)) -> cons(var(s(n)),genvar(n))

  minus(bot, appl(f,lv), sig) -> bot
  minus(plus(u,v), appl(f,lv), sig) -> 
        plus(minus(u,appl(f,lv), sig),
             minus(v,appl(f,lv), sig))
  minus(appl(f,lu), appl(g,lv), sig) -> appl(f,lu)

  minus(appl(f,lu), appl(f,lv), sig) -> 
        genm7(f,lu,lv,len(lu), sig)
  genm7(f,lu,lv,z, sig) -> bot
  genm7(f,lu,lv,suc(i), sig) -> 
        plus(genm7(f,lu,lv,i, sig), 
             appl(f,diff(lu,lv,suc(i), sig)))
  diff(nil,nil,i, sig) -> nil
  diff(cons(u,lu),cons(v,lv),s(s(i)), sig) -> 
        cons(u,diff(lu,lv,s(i), sig))
  diff(cons(u,lu),cons(v,lv),s(z), sig) -> 
        cons(minus(u,v, sig),lu)
  len(nil) -> z
  len(cons(u,lu)) -> s(len(lu))

  minus(at(var(n),v),w, sig) -> 
        at(var(n),minus(v,w, sig))
  minus(v,at(var(n),w), sig) -> minus(v,w, sig)
)
\end{verbatim}

\end{document}